\documentclass[letterpaper,11pt]{article}

\usepackage[Arxiv]{optional}
    
\usepackage{verbatim}

\usepackage[margin=1in]{geometry}
\usepackage{setspace}
\usepackage{graphicx}
\usepackage{epsfig}

\usepackage{microtype} 
\usepackage{array} 

\usepackage{amssymb}
\usepackage{amsmath}
\usepackage{amsthm}
\usepackage{mathtools} 
\usepackage{stmaryrd} 
\usepackage{stackrel} 
\usepackage{latexsym} 

\usepackage{dutchcal} 

\usepackage{tikz}
\usetikzlibrary{arrows}
\usetikzlibrary{arrows.meta}
\usetikzlibrary{decorations.pathreplacing}
\usetikzlibrary{cd}
\usepackage[all]{xy} 

\usepackage{algorithm}
\usepackage[noend]{algpseudocode}
\algnewcommand{\LineComment}[1]{\Statex\hspace{\algorithmicindent}\(\triangleright\) #1}
\algnewcommand\algorithmicforeach{\textbf{for each}}
\algdef{S}[FOR]{ForEach}[1]{\algorithmicforeach\ #1\ \algorithmicdo}
\usepackage{etoolbox} 

\usepackage{caption}
\usepackage{subcaption}

\usepackage{makecell} 
\usepackage{booktabs} 
\usepackage{multirow} 
\usepackage{tablefootnote} 

\usepackage[hidelinks]{hyperref} 
\usepackage[symbol]{footmisc}
\usepackage{makeidx} 
\usepackage{url} 
\usepackage{color}
\usepackage{xcolor}
\usepackage{xspace} 
\usepackage[normalem]{ulem} 
\usepackage{soul} 
\usepackage{ifpdf} 
\usepackage{extarrows}
\usepackage{stackengine}    
\usepackage{lipsum}
\usepackage{bm}
\usepackage{cleveref} 

\usepackage[style=numeric-comp, sorting=none]{biblatex}
\usepackage{enumitem}
\setlist[itemize]{leftmargin=\parindent}
\setlist[enumerate]{leftmargin=\parindent}
\setlist[description]{font=\bfseries,leftmargin=\parindent}

\makeatletter
\expandafter\patchcmd\csname\string\algorithmic\endcsname{\itemsep\z@}{\itemsep=0.25ex}{}{}
\makeatother

\theoremstyle{plain}
    \newtheorem{theorem}{Theorem}
    \newtheorem{proposition}[theorem]{Proposition}

    \newtheorem*{algr*}{Algorithm}
\theoremstyle{definition}
    \newtheorem{definition}{Definition}

    \newtheorem{remark}{Remark}
    \newtheorem*{remark*}{Remark}
    \newtheorem*{example*}{Example}

\newcommand{\revise}[1] {{{{#1}}}}

\title{A Fast Algorithm for Computing Zigzag Representatives}

\author{Tamal K. Dey\thanks{Department of Computer Science, Purdue University, West Lafayette, IN, USA. \texttt{tamaldey@purdue.edu}}
\and Tao Hou\thanks{Department of Computer Science, University of Oregon, Eugene, OR, USA. \texttt{taohou@uoregon.edu}}
\and  Dmitriy Morozov\thanks{Lawrence Berkeley National Laboratory, Berkeley, CA, USA. \texttt{dmitriy@mrzv.org}}
}

\date{}

\newcommand{\Hm}{\mathsf{H}}

\newcommand{\Chn}{\mathsf{C}}
\newcommand{\Zyc}{\mathsf{Z}}
\newcommand{\Bnd}{\mathsf{B}}

\renewcommand{\ker}{\mathsf{ker}}

\newcommand{\Pers}{\mathsf{Pers}}
\newcommand{\PersH}{\mathsf{Pers}^H}
\newcommand{\PersB}{\mathsf{Pers}^B}

\newcommand{\pinds}{\mathsf{P}}
\newcommand{\pindsH}{\mathsf{P}^H}
\newcommand{\pindsB}{\mathsf{P}^B}
\newcommand{\ninds}{\mathsf{N}}
\newcommand{\nindsH}{\mathsf{N}^H}
\newcommand{\nindsB}{\mathsf{N}^B}

\renewcommand{\bar}[1]{\overline{#1}}

\newcommand{\lbarrowspace}{\;}

\let\leftrightarrowsp\lrarrowsp

\newcommand{\incto}{\hookrightarrow}
\newcommand{\inctosp}[1]{\xhookrightarrow{\lbarrowspace#1\lbarrowspace}}
\newcommand{\bakincto}{\hookleftarrow}
\newcommand{\bakinctosp}[1]{\xhookleftarrow{\lbarrowspace#1\lbarrowspace}}

\newcommand{\given}{\,|\,}
\newcommand{\Set}[1]{\{#1\}}
\newcommand{\bigSet}[1]{\big\{#1\big\}}

\newcommand{\repsum}{\boxplus}

\newcommand{\cancelText}[1]{}

\let\emptyset\varnothing
\let\intersect\cap

\let\union\cup

\newcommand{\Fcal}{\mathcal{F}}

\newcommand{\Ical}{\mathcal{I}}

\newcommand{\Zbb}{\mathbb{Z}}

\newcommand{\aG}{\alpha}

\newcommand{\iG}{\iota}

\newcommand{\lG}{\lambda}
\newcommand{\LG}{\Lambda}

\newcommand{\oG}{\omega}

\newcommand{\sG}{\sigma}

\newcommand{\blue}{}

\newcommand{\birth}{b}
\newcommand{\bbirth}{b'}
\newcommand{\hbirth}{\hat{b}}
\newcommand{\death}{d}
\newcommand{\filtcnt}{m}
\newcommand{\simpcnt}{n}

\newcommand{\filt}{\Fcal}
\newcommand{\Filt}[1]{\Fcal_{#1}}
\newcommand{\filtlen}{\ell}
\newcommand{\cyc}{z}
\newcommand{\linmap}{\psi^*}
\newcommand{\linmapB}{\psi^{\#}}
\newcommand{\fsimp}{\sigma}

\newcommand{\wire}[1]{\oG_{#1}}
\newcommand{\bndl}{W}
\newcommand{\bndlsum}{\boxplus}

\newcommand{\bndbndl}{U}

\newcommand{\Ktot}{\bar{K}}

\newcommand{\indmap}{\chi}

\newcommand{\cycmat}{Z}
\newcommand{\bndmat}{B}
\newcommand{\chnmat}{C}

\newcommand{\pivot}{\mathrm{pivot}}
\newcommand{\matcol}[2]{#1[{#2}]}
\newcommand{\lessB}{\prec}
\newcommand{\cmatcol}[1]{\cycmat[{#1}]}

\newcommand{\defemph}[1]{\textbf{\textit{#1}}}

\newcommand{\Int}{\mathbb{I}}
\newcommand{\M}{\mathsf{M}}
\newcommand{\EM}{\overline{\mathsf{M}}}
\newcommand{\Zs}{\mathsf{Z}}
\newcommand{\Bs}{\mathsf{B}}
\newcommand{\Cs}{\mathsf{C}}
\newcommand{\cvec}{{\bf vec}_{\mathbb F}}
\newcommand{\rep}{\mathsf{rep}}

\makeatletter
\newcommand*{\da@rightarrow}{\mathchar"0\hexnumber@\symAMSa 4B }
\newcommand*{\da@leftarrow}{\mathchar"0\hexnumber@\symAMSa 4C }
\newcommand*{\xdashrightarrow}[2][]{%
  \mathrel{%
    \mathpalette{\da@xarrow{#1}{#2}{}\da@rightarrow{\;}{}}{}%
  }%
}
\newcommand{\xdashleftarrow}[2][]{%
  \mathrel{%
    \mathpalette{\da@xarrow{#1}{#2}\da@leftarrow{}{}{\;}}{}%
  }%
}
\newcommand{\xdashleftrightarrow}[2][]{%
  \mathrel{%
    \mathpalette{\da@xarrow{#1}{#2}\da@leftarrow\da@rightarrow{}{}}{}%
  }%
}
\newcommand*{\da@xarrow}[7]{%
  \sbox0{$\ifx#7\scriptstyle\scriptscriptstyle\else\scriptstyle\fi#5#1#6\m@th$}%
  \sbox2{$\ifx#7\scriptstyle\scriptscriptstyle\else\scriptstyle\fi#5#2#6\m@th$}%
  \sbox4{$#7\dabar@\m@th$}%
  \dimen@=\wd0 %
  \ifdim\wd2 >\dimen@
    \dimen@=\wd2 %
  \fi
  \count@=2 %
  \def\da@bars{\dabar@\dabar@}%
  \@whiledim\count@\wd4<\dimen@\do{%
    \advance\count@\@ne
    \expandafter\def\expandafter\da@bars\expandafter{%
      \da@bars
      \dabar@ 
    }%
  }%
  \mathrel{#3}%
  \mathrel{%
    \mathop{\da@bars}\limits
    \ifx\\#1\\%
    \else
      _{\copy0}%
    \fi
    \ifx\\#2\\%
    \else
      ^{\copy2}%
    \fi
  }%
  \mathrel{#4}%
  \!\!
}
\makeatother

\newcounter{desccounter}

\newcommand{\cancel}[1]

\let\defemph\emph

\begin{document}

\maketitle

\opt{Arxiv}{

\begin{abstract}
Zigzag filtrations of simplicial complexes generalize the usual filtrations 
by allowing simplex deletions in addition to simplex insertions.
The barcodes computed from zigzag filtrations encode the evolution of homological features.
Although one can locate a particular feature at any index in the filtration using existing algorithms,
the resulting \emph{representatives} may not be compatible with the zigzag: a representative cycle at one index may not map into a representative cycle at its neighbor.
For this, one needs to compute compatible representative cycles along each bar in the barcode. 
\revise{It is known that the barcode for a zigzag filtration with $m$ insertions and deletions can be computed in $O(m^\omega)$ time,
where $\omega< 2.373$ is the matrix multiplication exponent.
However,} it is not known how to compute the compatible representatives so efficiently. 
For a non-zigzag filtration,  the classical matrix-based algorithm provides representatives in $O(m^3)$ time, which can be improved to $O(m^\omega)$.
However, no known algorithm for zigzag filtrations  computes the representatives with the $O(m^3)$
time bound. 
We present an $O(m^2n)$ time algorithm for this problem, where $n\leq m$ is the size of the largest complex in the filtration.
\end{abstract}

}

\section{Introduction}
\label{sec:intro}

Persistent homology and its computation have been a central theme in topological data analysis (TDA)~\cite{DW22,edelsbrunner2010computational,Oudot15}. 
Using persistent homology,
one computes a signature called a \emph{barcode} from
data which is presented in the form of a growing sequence of simplicial
complexes called a \emph{filtration}. However, the barcode itself does not provide an avenue to go back to the data. For that, we need to compute
a representative for each \emph{bar} (interval) in the barcode, that is, a cycle 
whose homology class exists exactly over the duration of the bar. In other words,
we aim to compute the \emph{interval modules} themselves in the interval decomposition~\cite{Gabriel72} instead of only the intervals. 

In this paper, we consider computing representatives for the bars where the given
filtration is no longer monotonically growing but may also shrink, resulting in
what is known as a \emph{zigzag} filtration.
A number of algorithms have
been proposed for computing the barcode from a zigzag filtration~\cite{carlsson2009zigzag-realvalue,DBLP:conf/esa/DeyH22,DW22,maria2014zigzag,MS19,milosavljevic2011zigzag}.
All of them maintain \emph{pointwise representatives}, i.e., a homology
basis for every step in the filtration, but they do not compute the
\emph{barcode representatives}, i.e., a set of compatible pointwise bases, where
elements of one basis are matched to the elements of its neighbors (see Definition~\ref{dfn:rep-cls}).
Solving this problem is the main topic of this paper.

The barcode representatives are not readily available during the zigzag
computation because basis updates at any point may require
changes both in the future and in the past to maintain the matching.
To make this precise, let $m$ be the number of additions and deletions 
and $n$ be the maximum size of complexes in
a zigzag filtration.
The challenge is rooted
in the fact that a 
barcode representative for a zigzag filtration (henceforth also called a \emph{zigzag representative})
may consist of $O(m)$ different cycles~\cite{maria2014zigzag}
for each of the $O(m)$ indices
in a bar (see Definition~\ref{dfn:rep-cls}).
Consequently, the space complexity for the straightforward way of
maintaining a zigzag representative is $O(mn)$.
This is
in contrast to a non-zigzag representative which
consists of the same cycle over the entire bar.
One obvious way to obtain the zigzag representatives
is to adapt the $O(mn^2)$ algorithm proposed by Maria and Oudot~\cite{maria2014zigzag}
which directly targets representatives.
But then, the complexity increases to $O(m^2n^2)$,
which stems from the
need of summing two representatives each consisting of $O(m)$ cycles. In total
these summations over the entire course of the algorithm incur an $O(m^2n^2)$ cost.
{\blue To see this,
notice that the algorithm in~\cite{maria2014zigzag}
is based on summations of bars (and their representatives)
where each bar is associated with a single cycle from the $O(m)$ cycles in its representative.
The algorithm 
performs $O(mn)$ summations 
of bars and the associated cycles resulting in an $O(mn^2)$ complexity.
To adapt this algorithm for computing representatives,
one instead maintains the full representative consisting of $O(m)$ cycles
for each bar.
Because a summation of two bars now costs $O(mn)$ time,
the $O(mn)$ bar summations in the algorithm~\cite{maria2014zigzag}
then result in an $O(m^2n^2)$ complexity.}

It has remained tantalizingly difficult to design an algorithm that brings down the theoretical complexity to $O(m^3)$, matching the complexity for
non-zigzag filtrations~\cite{cohen2006vines,milosavljevic2011zigzag}, 
while remaining practical.
As mentioned already, the bottleneck of the computation lies in 
the summation of two representatives
each consisting of $O(m)$ cycles.
In this paper, we present an $O(m^2n)$ algorithm which overcomes the bottleneck
by compressing the representatives into a more compact form each taking
only $O(m)$ space instead of $O(mn)$ space.

\paragraph{Figure~\ref{fig:repchange}: an illustrative example.}
The compression of representatives in our algorithm is made possible by 
adopting some novel constructs for computing zigzag persistence whose ideas are illustrated in Figure~\ref{fig:repchange} 
(see also the beginning of Section~\ref{sec:rep-as-bndl} for more explanations;
formal definitions of concepts mentioned below are provided in Section~\ref{sec:core-def}):

\begin{figure}[!tbp]
  \centering
  \opt{Arxiv}{\includegraphics[width=0.8\linewidth]{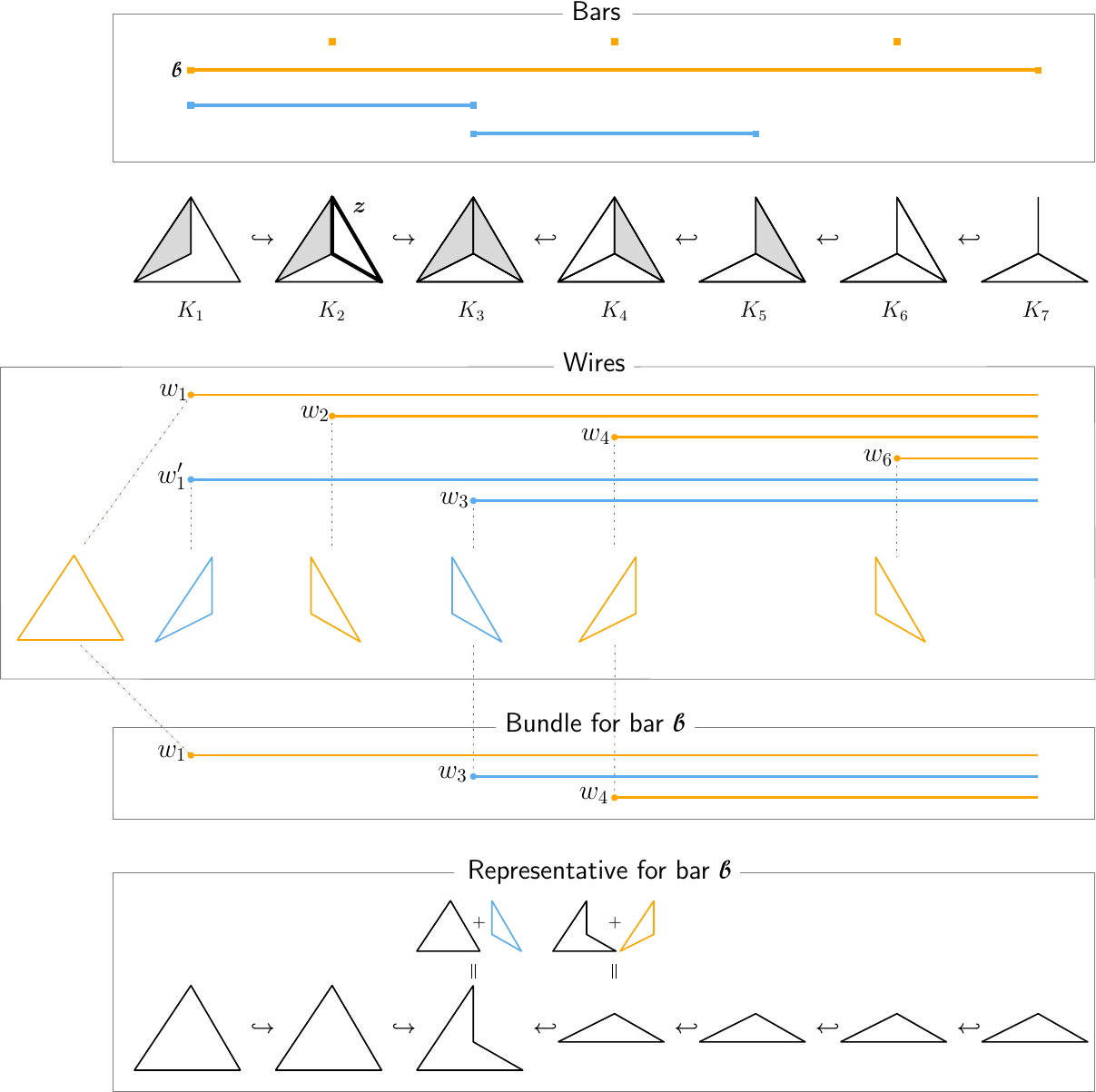}}
  \opt{SN}{\includegraphics[width=\linewidth]{fig/repchange}\bigskip}

  \medskip
  \caption{An example of wires, bundles, and the boundary zigzag module
  which are major constructs leading to the $O(\filtcnt^2\simpcnt)$ algorithm.
  Orange and blue colors are used for the constructs of
  homology and boundary zigzag modules respectively.}
  \label{fig:repchange}
\end{figure}

\begin{itemize}
    \item First, we observe that the barcode of the regular (homology) zigzag module interconnects with
the barcode of another module, namely, the \emph{boundary} zigzag module,
which arises out of the boundary groups for complexes in the input zigzag filtration.
To see the interconnection,
let $\bm{z}$ denote 
the bold cycle in $K_2$ (and its continuation
in the complexes $K_3$$-$$K_5$) in Figure~\ref{fig:repchange}.
In Figure~\ref{fig:repchange} (top), the bar $[2,2]$ 
for the homology module
born at $K_2$ and dying entering $K_3$ (hence drawn as an orange dot at index $2$) interconnects with the bar $[3,5]$ 
for the boundary module
born at $K_3$ and dying entering $K_6$ as the cycle $\bm{z}$
representing the bar $[2,2]$ becomes a boundary in $K_3$. The 
bar $[3,5]$,
which is also represented by $\bm{z}$,
in turn interconnects with the bar $[6,6]$ 
for the homology module
as $\bm{z}$ becomes a non-boundary at $K_6$.
\item
Second, we observe that
the seamless transition between 
barcodes of the two modules
allows us to define a construct
called \emph{wires} each of which is a single cycle  
with a fixed birth index,
presumably extending indefinitely to infinity. 
A wire
may be a boundary cycle 
(thus called a \emph{boundary} wire) 
with its birth index coinciding 
with a birth 
in the boundary module, or a non-boundary cycle 
(thus called a \emph{non-boundary} wire)
with its birth index coinciding with a birth 
in the 
homology module. For the example in Figure~\ref{fig:repchange},
we have \revise{four non-boundary wires} (orange) and two boundary wires (blue)
subscripted by the birth indices
with respective cycles also being illustrated.

A collection
of such wires forms what we call a \emph{bundle} for a zigzag bar. 
In Figure~\ref{fig:repchange}, we show the bundle for the longest bar $\mathbcal{b}=[1,7]$. 
One surprising fact we find is that
representative cycles of a bar can be recovered from  
index-wise summations of the wire cycles
in its bundle even though a
wire cycle involved in the summation may not be present in each complex over the bar (see Section~\ref{sec:rep-as-bndl}).
Figure~\ref{fig:repchange} (bottom) shows the representative cycles of the bar $\mathbcal{b}$ obtained by summing three wires $\{w_1,w_3, w_4\}$ even though the cycles for wires $w_1,w_4$ are not present in $K_5$$-$$K_7$.


At each index in the filtration, there can be no more than one wire with birth at that index. Hence,
each bundle is represented as a set of $O(m)$ wire indices
in our algorithm.
The summations among the bundles are then less costly
and can be done in $O(m)$ time
because each entails doing a symmetric sum among $O(m)$ wire indices 
rather than the actual 
$O(m)$ cycles.
When a bar is completed, its actual representative is read from summing
the cycles in its bundle. 
Wires and bundles allow our algorithm to have a space complexity of $O(mn)$ whereas the algorithm for computing representatives
adapted from~\cite{maria2014zigzag} has a space complexity of $O(mn^2)$.
\revise{Notice that the space complexity $O(mn)$ indicates
the working memory of the algorithm, while the output size
is $O(m^2)$ for $O(m)$ bundles.}
\end{itemize}

Our compression using wires is also made possible
by adopting
a new way
of computing zigzag barcodes, which processes the filtration 
from left to right similar to the algorithm in~\cite{carlsson2009zigzag-realvalue}
but directly targets maintaining the zigzag representatives
over the course of the computation. 
This is also in contrast to the other representative-based algorithm~\cite{maria2014zigzag} 
which
always maintains a reversed non-zigzag filtration
at the end.
Section~\ref{sec:rep-as-bndl} briefly describes the idea.

\section{Core definitions}\label{sec:core-def}

Throughout,
we assume a \emph{simplex-wise zigzag filtration} $\Fcal$ as input to our algorithm:
\begin{equation}
\label{eqn:prelim-filt}
\Fcal: \emptyset=K_0 \leftrightarrowsp{\fsimp_0} K_1 \leftrightarrowsp{\fsimp_1} 
\cdots \leftrightarrowsp{\fsimp_{\filtcnt-1}} K_\filtcnt,
\end{equation}
in which each $K_i$ is a simplicial complex and each arrow
$K_i \leftrightarrowsp{\fsimp_i} K_{i+1}$ is either a forward inclusion
$K_i \inctosp{\fsimp_i} K_{i+1}$ (an addition of a simplex $\fsimp_i$)
or a backward one $K_i \bakinctosp{\fsimp_i} K_{i+1}$ (a deletion of a simplex $\fsimp_i$).
Notice that assuming $\Fcal$ to be simplex-wise and $K_0=\emptyset$ 
is a standard practice in the computation
of non-zigzag persistence~\cite{edelsbrunner2000topological}
and its zigzag version~\cite{carlsson2009zigzag-realvalue,maria2014zigzag}.
Also notice that any zigzag filtration in general can be converted into a
simplex-wise version, and
the representatives computed for this simplex-wise version
can also be easily mapped to the ones for the original filtration.
We let $\Filt{i}$ denote the part of $\Fcal$ up to index $i$, that is, 
\begin{equation}
\label{eqn:partial-filt}
\Filt{i}: \emptyset=K_0 \leftrightarrowsp{\fsimp_0} K_1 \leftrightarrowsp{\fsimp_1} 
\cdots \leftrightarrowsp{\fsimp_{i-1}} K_i.
\end{equation}
Notice that $\Fcal=\Filt{\filtcnt}$.
The \emph{total complex} $\Ktot$ of $\Fcal$ 
is the union of all 
complexes in $\Fcal$.
Let $\simpcnt$ be the maximum size of  complexes in $\Fcal$
(note that generally $\simpcnt$ is not equal to the size of $\Ktot$).

{\blue In this paper, we consider homology under the field $\Zbb_2$
for the ease of presentation. The idea described can naturally be generalized to 
arbitrary field coefficients.}
For a complex $K_i$,
we consider its homology group $\Hm(K_i)$ 
{over all degrees},
which is the direct sum of $\Hm_p(K_i)$ for all $p$
({\blue so that the dimension of  $\Hm(K_i)$ equals the sum of the dimensions of all $\Hm_p(K_i)$'s}).
Accordingly,
$\Chn(K_i)$, $\Zyc(K_i)$, and $\Bnd(K_i)$ denote the chain, cycle,
and boundary groups of $K_i$ over all degrees respectively
\revise{(similar to $\Hm(K_i)$, dimensions of these spaces equal sums of the dimensions 
over all degrees $p$)}.
Since we take $\Zbb_2$ as coefficients,
chains or cycles in this paper are also treated
as sets of simplices.
We also consider any chain $c\in\Chn(K_i)$  to be
a chain in $\Ktot$ in general
and  do not differentiate the same simplex
appearing in different complexes in $\Fcal$. 
For example, 
suppose that all simplices in $c\in\Chn(K_i)$
also belong to a $K_j$, 
we then have $c\in\Chn(K_j)$.

Taking the homology functor on $\Fcal_i$ we obtain the following (homology) \emph{zigzag module}:
\begin{equation*}
\Hm(\Fcal_i): 
\Hm(K_0) 
\leftrightarrowsp{\linmap_0} 
\Hm(K_1) 
\leftrightarrowsp{\linmap_1} 
\cdots 
\leftrightarrowsp{\linmap_{\filtcnt-1}} 
\Hm(K_i).
\end{equation*}

Similarly,  {\blue taking the boundary groups of complexes in $\Fcal_i$
and the chain maps between them, which indeed defines a \emph{boundary functor},}
we obtain the  (boundary) {zigzag module}:
\[\Bnd(\Fcal_i): 
\Bnd(K_0) 
\leftrightarrowsp{\linmapB_0} 
\Bnd(K_1) 
\leftrightarrowsp{\linmapB_1} 
\cdots 
\leftrightarrowsp{\linmapB_{\filtcnt-1}} 
\Bnd(K_i). \]
Each 
$\linmap_j:\Hm(K_j)\leftrightarrow \Hm(K_{j+1})$
in $\Hm(\Fcal_i)$
is a linear map induced by inclusion between homology groups whereas 
each 
$\linmapB_j:\Bnd(K_j)\leftrightarrow \Bnd(K_{j+1})$ 
in $\Bnd(\Fcal_i)$ is
an inclusion between chain  groups.
By~\cite{carlsson2010zigzag,Gabriel72}, for some index sets $\LG_H$ and
$\LG_B$,
$\Hm(\Fcal_i)$ and $\Bnd(\Fcal_i)$
have decompositions of the form
\[
\Hm(\Fcal_i)=\bigoplus_{k\in\LG_H}\Ical^{[\birth_k,\death_k]}\quad\text{ and }\quad
\Bnd(\Fcal_i)=\bigoplus_{k\in\LG_B}\Ical^{[\birth_k,\death_k]},\]
in which each $\Ical^{[\birth_k,\death_k]}$
is an
{\it interval module} over the interval $[\birth_k,\death_k]\subseteq\Set{0,1,\ldots,i}$.
The 
set of intervals
$\Pers^H(\Fcal_i):=\Set{[\birth_k,\death_k]\given k\in\LG_H}$ for $\Hm(\Fcal_i)$
and the set of intervals $\Pers^B(\Fcal_i):=\Set{[\birth_k,\death_k]\given k\in\LG_B}$ for $\Bnd(\Fcal_i)$ 
are called the {\it  homology barcode} 
and {\it boundary barcode} of $\Fcal_i$ respectively.
In this paper, we introduce the computation 
of the intervals and representatives for $\Bnd(\Fcal)$
 as an integral part of the computation of those for $\Hm(\Fcal)$,
which is critical to achieving the $O(\filtcnt^2\simpcnt)$ complexity.
We similarly define a barcode $\Pers_p^H(\Fcal_i)$ for the module $\Hm_p(\Fcal_i)$ over each degree $p$,
so that $\Pers^H(\Fcal_i)=\bigsqcup_p \Pers^H_p(\Fcal_i)$.
Notice that we can also define the barcode  $\Pers_p^B(\Fcal_i)$ 
where $\Pers^B(\Fcal_i)=\bigsqcup_p \Pers^B_p(\Fcal_i)$.

\begin{definition}[Homology birth/death indices]
\label{dfn:pos-neg-inds}
Since $\Filt{i}$ is simplex-wise,
each map $\linmap_j$
in $\Hm(\Filt{i})$
is either injective with a 1-dimensional cokernel
or  surjective with a 1-dimensional kernel
but cannot be both.
The set of \defemph{homology birth indices} of $\Filt{i}$,
denoted $\pinds^H(\Filt{i})$,
and the {\blue(multi-)}set of \defemph{homology death indices} of $\Filt{i}$,
denoted $\ninds^H(\Filt{i})$,
are  constructively defined as follows:
for each forward 
$\linmap_j:\Hm(K_j)\to \Hm(K_{j+1})$,
we {\blue either include} 
$j+1$ to $\pinds^H(\Filt{i})$
if $\linmap_j$
is injective or
include 
$j$ to $\ninds^H(\Filt{i})$ otherwise.
Also,
for each backward  
$\linmap_j:\Hm(K_j)\leftarrow \Hm(K_{j+1})$,
we {\blue either include}
$j+1$ to $\pinds^H(\Filt{i})$
if $\linmap_j$
is surjective or include
$j$ to $\ninds^H(\Filt{i})$ otherwise.
Finally, we {\blue include} $r$ copies of 
$i$ to $\ninds^H(\Filt{i})$
where $r$ is the dimension of $\Hm(K_i)$.
\end{definition}

\begin{definition}[Boundary birth/death indices]
Similarly as above, we define the \defemph{boundary birth indices} $\pinds^B(\Filt{i})$
and \defemph{boundary death indices} $\ninds^B(\Filt{i})$ of $\Filt{i}$ by considering
the 
module $\Bnd(\Filt{i})$.
Notice that $\linmapB_j$ is always injective.
So, for each forward 
$\linmapB_j:\Bnd(K_j)\to \Bnd(K_{j+1})$
that is not surjective,
we include $j+1$ to $\pindsB(\Filt{i})$.
Also,
for each backward  
$\linmapB_j:\Bnd(K_j)\leftarrow \Bnd(K_{j+1})$
that is not surjective,
we include $j$ to $\nindsB(\Filt{i})$.
Finally, we include $q$ copies of $i$ to $\ninds^B(\Filt{i})$
where $q$ is the dimension of $\Bnd(K_i)$.
\end{definition}
 Whenever $\linmap_j$ is injective, $\linmapB_j$ is an identity map;
whenever  $\linmap_j$ is surjective, $\linmapB_j$ is not surjective.
Hence, $\pindsH(\Filt{i})\intersect\pindsB(\Filt{i})=\emptyset$
while (different copies of) $i$ could belong to both $\nindsH(\Filt{i})$ and $\nindsB(\Filt{i})$.
Also notice that $[b,d]\in\PersH(\Filt{i})$
implies that $b\in\pindsH(\Filt{i})$ and $d\in\nindsH(\Filt{i})$
(similar facts hold for $[b,d]\in\PersB(\Filt{i})$).
We provide the definition of  
homology representatives 
(see
Maria and Oudot~\cite{maria2014zigzag}) as follows and then
adapt it to define boundary representatives:

\begin{definition}[Homology representatives]
\label{dfn:rep-cls}
Consider a filtration $\Filt{i}$
and 
let $[b,d]\subseteq[0,i]$ be an interval 
where $b\in\pinds^H(\Filt{i})$ 
(notice that $b> 0$ because $K_0=\emptyset$ by assumption)
and $d\in\ninds^H(\Filt{i})$.
A sequence of cycles
$\rep=\Set{\cyc_\aG\in \Zyc(K_\aG) \given \aG\in[b,d]}$ 
is called a 
\defemph{homology representative}
(or simply \defemph{representative})
for 
$[b,d]$ if for every $b\leq \aG<d$, 
either $\linmap_\aG([z_\aG])=[z_{\aG+1}]$ or
$\linmap_\aG([z_{\aG+1}])=[z_{\aG}]$ based on the direction of $\linmap_\aG$.
Furthermore, we have:
\begin{description}
    \item[Birth condition:] 
    If $\linmap_{b-1}:\Hm(K_{b-1})\to\Hm(K_{b})$ is forward
    (thus being injective),
    $\cyc_b\in\Zyc(K_b)\setminus \Zyc(K_{b-1})$;
    if $\linmap_{b-1}:\Hm(K_{b-1})\leftarrow\Hm(K_{b})$ is backward
    (thus being  surjective),
    then $[\cyc_b]$ is the non-zero element in $\ker(\linmap_{b-1})$.
    
    \item[Death condition:]
    If $d<i$ and
    $\linmap_{d}:\Hm(K_{d})\leftarrow\Hm(K_{d+1})$ is backward
    (thus being injective),
    $\cyc_d\in\Zyc(K_{d})\setminus \Zyc(K_{d+1})$;
    if $d<i$
    and $\linmap_{d}:\Hm(K_{d})\to\Hm(K_{d+1})$ is forward
    (thus being surjective),
    then $[\cyc_d]$ is the non-zero element in $\ker(\linmap_{d})$.
\end{description}
\end{definition}
\begin{remark}
By definition, all $z_\aG$'s in a homology representative $\rep$
are  $p$-cycles for the same $p$, 
so we can also call $\rep$  a $p$-th homology representative.
\end{remark}


\begin{definition}[Boundary representatives]
\label{dfn:brep-cls}
Let $[b,d]\subseteq[0,i]$ be an interval 
where $b\in\pinds^B(\Fcal_i)$ and $d\in\ninds^B(\Fcal_i)$.
A sequence of cycles
$\rep=\Set{\cyc_\aG\in \Bnd(K_\aG) \given \aG\in[b,d]}$ 
is called a 
\defemph{boundary representative}
(or simply \defemph{representative})
for the interval $[b,d]$ if for every $b\leq \aG<d$, 
either $z_{\aG+1}=\linmapB_\aG(z_\aG)\stackrel{\text{def}}{=}z_\aG$
or
$z_\aG=\linmapB_\aG(z_{\aG+1})\stackrel{\text{def}}{=}z_{\aG+1}$
based on the direction of $\linmapB_\aG$.
Furthermore, we have:
\begin{description}\label{dfn:brep}
    \item[Birth condition:] 
    The cycle $\cyc_b$ satisfies that $\cyc_b\in\Bnd(K_{b})\setminus\Bnd(K_{b-1})$
    where $\linmapB_{b-1}:\Bnd(K_{b-1})\to\Bnd(K_{b})$ 
    is forward
    because $b\in\pindsB(\Filt{i})$.
    \item[Death condition:]
    If $d<i$,
    then $\cyc_d$ satisfies that $\cyc_d \in \Bnd(K_d)\setminus \Bnd(K_{d+1})$
    where the map $\linmapB_{d}:\Bnd(K_{d})\leftarrow\Bnd(K_{d+1})$ 
    is backward 
    because $d\in\nindsB(\Filt{i})$.
\end{description}
\end{definition}

\begin{remark}
In the sequence $\rep$ in  Definitions~\ref{dfn:rep-cls} and~\ref{dfn:brep-cls},
we also call $z_\aG$  a cycle \emph{at} index $\aG$.
\end{remark}

The following Proposition 
is used later for proofs and algorithms.
\begin{proposition}
    Let $\cyc^B_1,\ldots,\cyc^B_k$ be the cycles at index $j$ in representatives for all intervals 
    of
    $\PersB(\Fcal_i)$ containing $j$. Similarly, let
    $\cyc_1^H,\ldots,\cyc_{k'}^H$ be the cycles at index $j$ in representatives for all intervals 
    of $\PersH(\Fcal_i)$ containing $j$.
    Then, 
    $\bigSet{[\cyc_1^H],\ldots,[\cyc_{k'}^H]}$
    is a basis of $\Hm(K_j)$,
    $\bigSet{\cyc^B_1,\ldots,\cyc^B_{k}}$
    is a basis of $\Bnd(K_j)$,
    and
    $\bigSet{\cyc_1^H,\ldots,\cyc_{k'}^H,\cyc^B_1,\ldots,\allowbreak\cyc^B_{k}}$
    is a basis of $\Zyc(K_j)$.
    \label{prop:cyclebasis}
\end{proposition}

\begin{proof}
    First, 
    the fact that
    $\bigSet{[\cyc_1^H],\ldots,[\cyc_{k'}^H]}$
    is a basis of $\Hm(K_j)$ and
    $\bigSet{\cyc^B_1,\ldots,\cyc^B_{k}}$
    is a basis of $\Bnd(K_j)$
    follows from the definition of interval decomposition and representatives.
    Consider any
    cycle $\cyc$ in $\Zyc(K_j)$. Then, there exists a unique $\alpha_t\in \{0,1\}$
    {\blue for each
    $1\leq t\leq k'$} so that $[\cyc]=\sum_t\alpha_t[\cyc^H_t]$. Then,
    $[\cyc]+\sum_t\alpha_t[\cyc^H_t]=[\cyc+\sum_t\alpha_t\cyc^H_t]=0$.
    It follows that $(\cyc+\sum_t\alpha_t\cyc^H_t)\in \Bnd(K_j)$, which implies that there exists a unique $\beta_\ell$ 
    for each $1\leq \ell\leq k$ so that
    $\cyc+\sum_t\alpha_t\cyc^H_t=\sum_\ell \beta_\ell\cyc^B_\ell$.
    So,
    $\cyc= \sum_t\alpha_t\cyc^H_t + \sum_\ell \beta_\ell \cyc^B_\ell$ for
    unique {\blue$\alpha_t$'s}, $1\leq t\leq k'$ and $\beta_\ell$'s, $1\leq \ell \leq k$.
    It follows that the union of the cycles $\{\cyc^B_\ell\}$ and $\{\cyc^H_t\}$
    generate $\Zyc(K_j)$. Since $k+k'=\mathrm{dim}(\Bnd(K_j)) + \mathrm{dim}(\Hm(K_j))=\mathrm{dim}(\Zyc(K_j))$, they form a basis.
\end{proof}



We then define 
summations of representatives for
intervals ending at $i$. These summations
respect a total order `$\lessB$' on birth indices~\cite{maria2014zigzag}, that is, a representative for $[b,i]$ 
can be added to a representative for $[b',i]$ if and only if $b\lessB b'$ (see Figure~\ref{fig:rep-sum}).

\begin{figure}[!tbp]
  \centering
  \includegraphics[width=\linewidth]{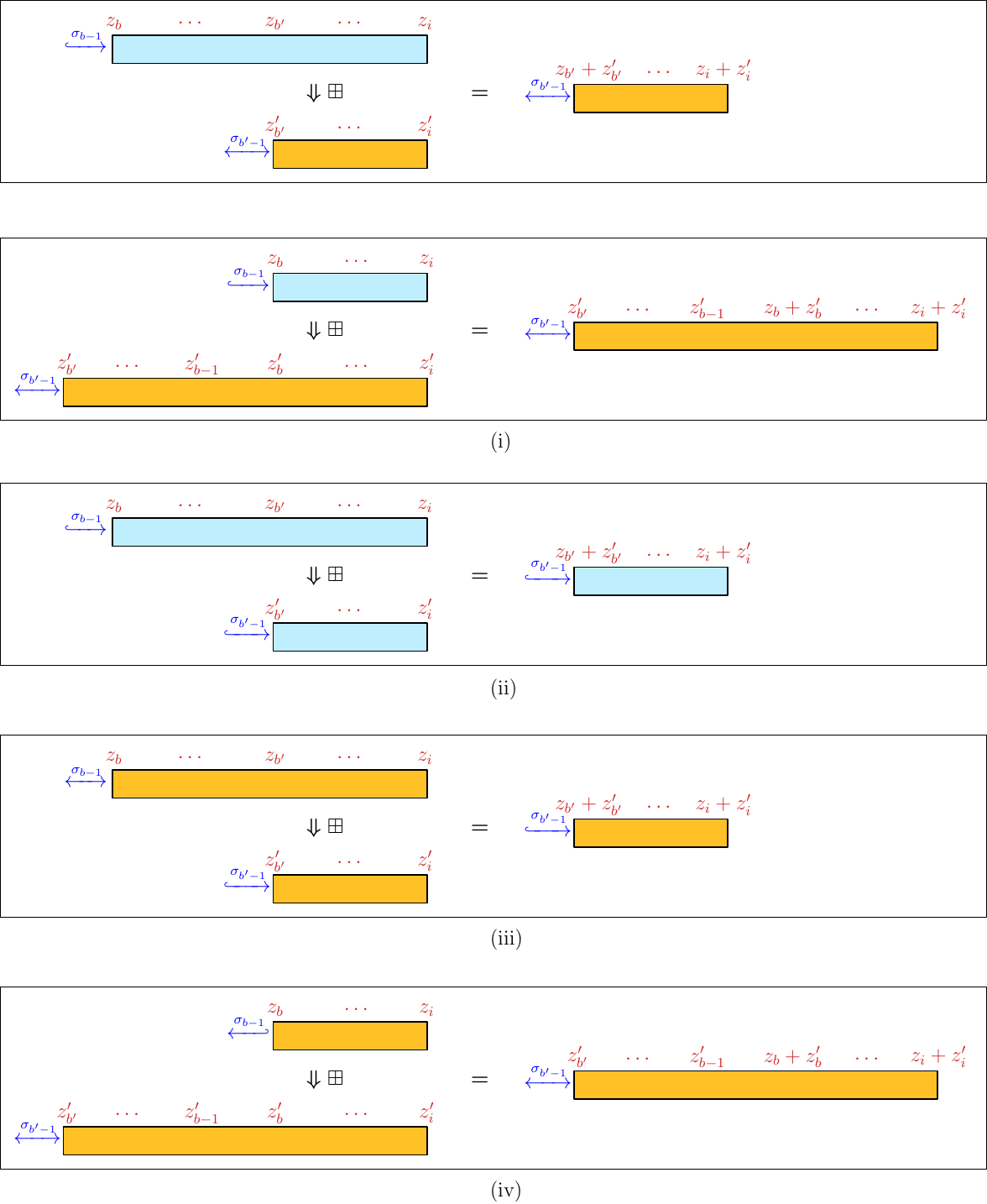}

  \opt{SN}{\bigskip}
  
  \caption{Illustration of how summations of representatives for
  intervals respect the
  order `$\lessB$' for the different cases in Definition~\ref{dfn:ind-order},
  with the double arrows indicating the directions of the summations.
  Boundary module intervals are shaded blue while homology module intervals are shaded orange.
  }
  \label{fig:rep-sum}
  \end{figure}

\begin{definition}[Total order on birth indices]
\label{dfn:ind-order}
For two birth indices $b,b'\in \pinds^H(\Fcal_i)\cup \pinds^B(\Fcal_i)$, 
we have $b\lessB b'$ if
one of the following holds: 
    \begin{enumerate}[label=(\roman*),leftmargin=2.2em]
        \item 
        $b\in \pinds^B(\Fcal_i)$ and $b'\in \pinds^H(\Fcal_i)$;
        \item 
        $b,b'\in \pinds^B(\Fcal_i)$ and $b<b'$;
       \item 
       $b,b'\in \pinds^H(\Fcal_i)$,  $b<b'$,  and $K_{b'-1}\inctosp{}K_{b'}$ is a forward inclusion;
       \item 
       $b,b'\in \pinds^H(\Fcal_i)$, $b'< b$,  and $K_{b-1}\bakinctosp{} K_{b}$ is a backward inclusion.
    \end{enumerate}
\end{definition}

\begin{definition}[Representative summation]
\label{dfn:rep-sum}
For two  intervals $[b,i],[b',i]\in \Pers_p^H(\Fcal_i)\cup\Pers_p^B(\Fcal_i)$
so that $b\lessB b'$,
let $\rep=\Set{z_\aG\mid \aG\in [b,i]}$ and
$\rep'=\Set{z'_\aG\mid \aG\in [b',i]}$ be $p$-th representatives
for $[b,i]$ and $[b',i]$
respectively.
The \defemph{sum} of $\rep$ and $\rep'$, denoted 
$\rep\repsum\rep'$, is a sequence of cycles $\Set{\bar{z}_\aG\mid\aG\in[b',i]}$,
for the interval $[b',i]$,
so that
\begin{itemize}
    \item 
    If $b<b'$ 
    then $\bar{z}_\aG=z_\aG+z'_\aG$ for each $\aG$; (Figure~\ref{fig:rep-sum}: (i) top, (ii), (iii))
    \item 
    If $b'<b$,
    then $\bar{z}_\aG=z'_\aG$ for $\aG<b$
    and $\bar{z}_\aG=z_\aG+z'_\aG$ for $\aG\geq b$. (Figure~\ref{fig:rep-sum}: (i) bottom, (iv))
\end{itemize}
\end{definition}
\revise{
\begin{remark}It is not hard to see that summing two homology (resp.\ boundary) representatives
always results in a homology (resp.\ boundary) representative.
On the other hand, summing a boundary representative with a homology representative 
results in a homology representative 
because summing a boundary with a cycle does not change the homology class of 
the cycle.
\end{remark}}


\begin{proposition}
    The sequence $\rep\repsum\rep'$ in Definition~\ref{dfn:rep-sum}
    is a $p$-th representative for
    $[b',i]\in \Pers_p^H(\Fcal_i)\cup\Pers_p^B(\Fcal_i)$. 
    \label{prop:repsum}
\end{proposition}
\begin{proof}
\noindent
    {\bf Case 1}, $b<b'$: In this case, every cycle $\bar{\cyc}_\alpha$,
    $\alpha\in[b',i]$, in
    $\rep\bndlsum\rep'$ satisfies that $\bar{\cyc}_\alpha=\cyc_\alpha+\cyc'_\alpha$. 
    It can be verified that we only have three different cases: (i) $b\in \pinds^B(\Fcal_i)$, $b'\in \pinds^H(\Fcal_i)$, (ii) $b\in \pinds^B(\Fcal_i)$, $b'\in \pinds^B(\Fcal_i)$, and (iii) 
    $b\in \pinds^H(\Fcal_i)$, $b'\in \pinds^H(\Fcal_i)$.
    We take up the case (i) and the proof for the other cases is similar.
    Assuming $K_\alpha\inctosp{}K_{\alpha+1}$ is  forward
    for $b'\leq \aG<i$,
    we have $\linmapB_\alpha(\cyc_\alpha)=\cyc_\alpha=\cyc_{\alpha+1}$ and
    $\linmap_\alpha([\cyc'_\alpha])=[\cyc'_{\alpha+1}]$. Therefore,
    $\linmap_\alpha([\bar{\cyc}_\alpha])=
    \linmap_\alpha([\cyc_\alpha])+\linmap_\alpha([\cyc'_\alpha])=
    [\cyc_{\alpha+1}]+[\cyc'_{\alpha+1}]= [\bar{\cyc}_{\alpha+1}]$ as required. The same applies
     when $K_\alpha\bakinctosp{}K_{\alpha+1}$ is backward. Finally, we verify that the birth and death conditions
    hold for $\bar{\cyc}_{b'}$.
    First assume that $K_{b'-1}\inctosp{} K_{b'}$ is forward. 
    For the birth condition,
    we have that $\bar{\cyc}_{b'}=\cyc_{b'}+\cyc'_{b'}\in \Zyc(K_{b'})\setminus \Zyc(K_{b'-1})$ because $\cyc'_{b'}\in \Zyc(K_{b'})\setminus \Zyc(K_{b'-1})$ and $\cyc_{b'}=\cyc_{b'-1}\in \Zyc(K_{b'})\cap \Zyc(K_{b'-1})$.
     One can also verify the death condition for the cycle
    $\bar{\cyc}_i$. 
    For a backward $K_{b'-1}\bakincto K_{b'}$, 
    we omit the verification for the birth and death conditions.
    It follows that in case (i), $\rep\bndlsum \rep'$ is
    a homology representative for $[b',i]$, $b'\in \pinds^H(\Fcal_i)$, as required. 
    We also have that the justification for case (ii) and (iii) can be similarly done.

    \smallskip
    \noindent
    {\bf Case 2}, $b' < b$: In this case, we have $\bar{\cyc}_\alpha=\cyc'_\alpha$ for $\alpha\in [b',b-1]$ and
    $\bar{\cyc}_\alpha=\cyc_\alpha + \cyc'_\alpha$ for $\alpha \in [b,i]$. 
    We have only two possible cases: (i) $b\in \pinds^B(\Fcal_i)$ and
    $b'\in \pinds^H(\Fcal_i)$; (ii) $b,b'\in \pinds^H(\Fcal_i)$ and $K_{b-1}\bakinctosp{}K_b$ is backward. Again, using the case analysis, one can show that $\linmap_\alpha([\bar{\cyc}_\alpha])=[\bar{\cyc}_{\alpha+1}]$ if $\linmap_\alpha$ is forward and $\linmap_\alpha([\bar{\cyc}_{\alpha+1}])=[\bar{\cyc}_\alpha]$ otherwise.
    Moreover, the birth and death conditions can also be verified easily implying that $\rep\bndlsum \rep'$ in both cases is a homology representative for $[b',i]$, $b'\in \pinds^H(\Fcal_i)$.
\end{proof}
\begin{remark}
From Figure~\ref{fig:rep-sum},
it is not hard to see that the representative resulting from the summation
in Definition~\ref{dfn:rep-sum} is still a valid representative for the interval.
For example, in case (iii) of Figure~\ref{fig:rep-sum},
the resulting representative is  valid because $z_{b'}+z'_{b'}$
still contains $\sG_{b'-1}$ so that the birth condition in Definition~\ref{dfn:rep-cls}
still holds.
\end{remark}




We then define wires and bundles as mentioned in Section~\ref{sec:intro}
which compresses the zigzag representatives in a compact form.


\begin{definition}[Wire]
\label{dfn:wire}
A \defemph{wire} is a cycle $\wire{i}\in\Zyc(K_i)$ 
with a \defemph{starting index} $i\in \pinds^H(\Fcal)\cup \pinds^B(\Fcal)$ 
such that
\begin{enumerate}[label=(\roman*),leftmargin=2.2em]
    \item 
    $K_{i-1}\inctosp{} K_i$ is  forward  and $\wire{i}\in \Zyc(K_i)\setminus \Zyc(K_{i-1})$, or
    \item 
    $K_{i-1}\bakinctosp{} K_i$ is backward  and $\wire{i} \in \Bnd(K_{i-1})\setminus \Bnd(K_i)$, or
    \item 
    $K_{i-1}\inctosp{} K_i$ is forward  and $\wire{i}\in \Bnd(K_i)\setminus \Bnd(K_{i-1})$.
\end{enumerate} 
We also say that  $\wire{i}$ is a wire \defemph{at} index $i$.
The wires satisfying (i) or (ii) are also called \defemph{non-boundary wires}
whereas those satisfying (iii) are called \defemph{boundary wires}. 
\end{definition}
\begin{remark}
    In cases (i) and (ii) above, $i\in \pinds^H(\Fcal)$, whereas in case (iii), $i\in \pinds^B(\Fcal)$.
\end{remark}


\begin{definition}[Wire bundle]
A \defemph{wire bundle} $\bndl$ (or simply \defemph{bundle}) is a set 
of wires with distinct starting indices.
The \defemph{sum} of $\bndl$ with another wire bundle $\bndl'$,
denoted $\bndl\bndlsum\bndl'$,
is the symmetric difference of the two sets.
We also call $\bndl$ a \defemph{boundary bundle}
if $\bndl$ contains only boundary wires and
call $\bndl$ a \defemph{non-boundary bundle} otherwise.
\end{definition}

As evident later,
given an input filtration $\Fcal$,
a wire at an index $i$ is  {fixed} in our algorithm,
and we always denote such a wire 
as $\wire{i}$.
Hence, a wire bundle 
 is simply stored as a list of wire indices 
in our algorithm.
Since there are $O(\filtcnt)$ indices in $\Fcal$,
a bundle summation takes $O(\filtcnt)$ time.
\begin{definition}
    Let $[b,d]\in \Pers^H(\Fcal_i)\cup\Pers^B(\Fcal_i)$. A wire bundle
    $W$ is said to \defemph{generate a representative} for $[b,d]$ $($or simply
    \defemph{represents} $[b,d]$$)$ if the sequence of cycles $\{ \cyc_\alpha=\sum_{\wire{j}\in W, j\leq \alpha}\wire{j}\mid \alpha\in [b,d]\}$ is a representative for $[b,d]$.
\end{definition}

\begin{remark}
In the sum 
$\cyc_\alpha=\sum_{\wire{j}\in W, j\leq \alpha}\wire{j}$
in the above definition,
we consider each $\wire{j}$ 
and the sum $z_\aG$
as a cycle in the total complex $\Ktot$.
Notice that 
if $\bndl$ 
generates a representative for
$[b,d]$,
we may have that $\wire{j}\not\in \Zyc(K_\aG)$
for a $\wire{j}$ in the sum, but
we still can have
$z_\aG\in \Zyc(K_\aG)$
due to cancellation of simplices in the symmetric difference.
See Figure~\ref{fig:repchange}.
\revise{Also notice that a wire in a bundle could start before the 
interval that the bundle represents. See the interval on the right 
in Figure~\ref{fig:bnd-sum}.}
\end{remark}

\revise{Figures~\ref{fig:repchange} and~\ref{fig:bnd-sum}} provide examples for representatives
generated by bundles.
Notice that since we always consider bundles that generate representatives in this paper,
bundle summations also respect the order `$\lessB$' in Definition~\ref{dfn:ind-order}. 
The main benefits of introducing wire bundles are that  (i) they can be summed efficiently and (ii) 
\revise{explicit representatives can also be generated from them efficiently}
(see the Algorithm \textsc{ExtRep} below for the detailed process).

\begin{algr*}[\textsc{ExtRep}: Extracting representative from bundle]\label{alg:bndl-gen}
Let $\bndl=\Set{\wire{\iG_1},\ldots,\wire{\iG_\ell}}$
be a wire bundle where $\iG_1<\cdots<\iG_\ell$
and let $\rep$ be the representative for an interval $[b,d]$ generated by $\bndl$.
We can assume $\iG_\ell\leq d$ because wires
in $W$ with indices greater than $d$ do not contribute to a cycle in $\rep$.
Moreover, let $\iG_k$ be the last index in $\iG_1,\ldots,\iG_\ell$
no greater than $b$.
We have that $z=\sum_{j=\iG_1}^{\iG_k}\wire{j}$
is the cycle at indices $[b,\iG_{k+1})$ in $\rep$.
We then let $\lG$ iterate over $k+1,\ldots,\ell-1$.
For each $\lG$, we add $\wire{\iG_{\lG}}$ to $z$,
and the resulting $z$ is the cycle at indices $[\iG_{\lG},\iG_{\lG+1})$ in $\rep$.
Finally, we
add $\wire{\iG_{\ell}}$ to $z$,
and $z$ is the cycle at indices $[\iG_{\ell},d]$ in $\rep$.
Since at every $\lG\in[k+1,\ell]$, we add at most one cycle to another cycle, the whole process involves $O(m)$ chain summations.
\end{algr*}



\section{Representatives as wire bundles}
\label{sec:rep-as-bndl}

We first give a brief overview of  our algorithm 
to illustrate 
how representatives in zigzag modules can be compactly stored as
wire bundles
(see Section~\ref{sec:alg-oview} for details of the algorithm).
Consider computing only the homology barcode $\PersH(\Fcal)$.
Our algorithm in Section~\ref{sec:alg-oview}  stems from 
an idea 
for computing $\PersH(\Fcal)$
that directly maintains representatives for the intervals:
Before each iteration $i$,
assume that we are given intervals in $\PersH(\Filt{i})$ and their representatives.
The aim of iteration $i$ is to compute those for $\PersH(\Filt{i+1})$
by processing the inclusion $K_i\leftrightarrowsp{\fsimp_i} K_{i+1}$.
For the computation,
we only need to pay attention to those \emph{active}
intervals in $\PersH(\Filt{i})$
ending with $i$ 
because the non-active intervals and their representatives have already been finalized.
Consider an interval $[b,i]\in\PersH(\Filt{i})$ with a representative $\rep$. {\blue If the cycle $\cyc_i$ at index $i$} in $\rep$ resides in $K_{i+1}$,
the interval $[b,i]\in\PersH(\Filt{i})$ can be directly extended 
to $[b,i+1]\in\PersH(\Filt{i+1})$ along with the representative where
the cycle at $i+1$ equals $z_i$.
Otherwise, if $z_i\nsubseteq K_{i+1}$,
we perform summations on the representatives 
to modify $\cyc_i$ in $\rep$  
so that $z_i$ becomes contained in $K_{i+1}$
and $[b,i]$ can be extended.

In iteration $i$, whenever the inclusion $K_i\leftrightarrow K_{i+1}$ 
generates a new birth index $i+1\in\pindsH(\Filt{i+1})$,
we have a new active interval $[i+1,i+1]\in\PersH(\Filt{i+1})$.
We assign a representative $\rep^{i+1}=\Set{z_{i+1}}$  to $[i+1,i+1]$
where $z_{i+1}$ only needs to satisfy the birth condition in Definition~\ref{dfn:rep-cls}.
Suppose that $[i+1,i+1]\in\PersH(\Filt{i+1})$
is directly extended
to $[i+1,k]\in\PersH(\Filt{k})$
in later iterations
without its representative $\rep^{i+1}=\Set{z_\aG\mid\aG\in[i+1,k]}$
being modified by representative summations.
We then have that $z_\aG=z_{i+1}$ for each $\aG$,
which means that $\rep^{i+1}$ is generated
by the wire $\wire{i+1}:=z_{i+1}$
(see  Figure~\ref{fig:simp-wire}.).
Suppose that we have  a similar 
interval $[j+1,k]\in\PersH(\Filt{k})$
with a representative $\rep^{j+1}$ also generated by 
a single wire $\wire{j+1}$,
where $j+1>i+1$ and $K_{j}\bakincto K_{j+1}$ is backward.
Then, $j+1\lessB i+1$ according to Definition~\ref{dfn:ind-order},
and we can
sum $\rep^{j+1}$ to $\rep^{i+1}$
to get a new representative for $[i+1,k]$. 
We have that 
the new representative is  generated by the bundle $\Set{\wire{i+1},\wire{j+1}}$
as illustrated in Figure~\ref{fig:simp-wire}.

\begin{figure}[!tb]
  \centering
  \includegraphics[width=\linewidth]{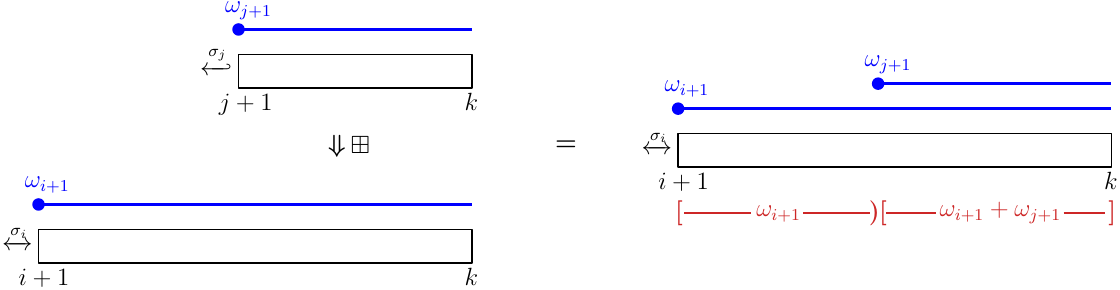}\opt{SN}{\bigskip}
  \caption{Summing two representatives each generated by a single wire
  results in a new representative generated by a bundle containing the two wires.}
  \label{fig:simp-wire}
  \end{figure}
\begin{figure}[htb]
  \centering
  \includegraphics[width=\linewidth]{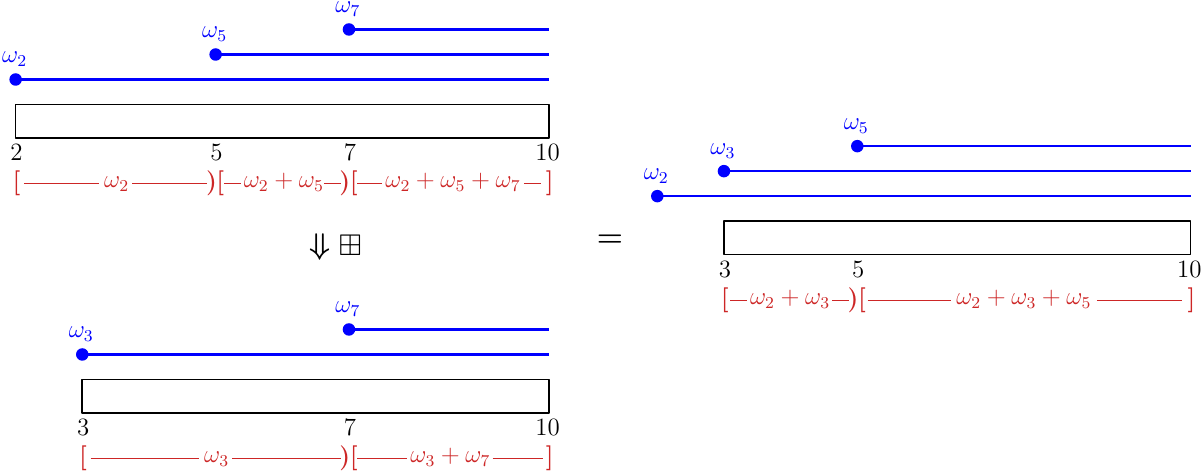}\opt{SN}{\bigskip}
  \caption{Summing two representatives generated by the bundles 
  $\Set{\wire{2},\wire{5},\wire{7}},\Set{\wire{3},\wire{7}}$
  respectively
  results in a new representative generated by the bundle $\Set{\wire{2},\wire{3},\wire{5}}$.}
  \label{fig:bnd-sum}
  \end{figure}
In the computation of $\PersH(\Fcal)$,
a representative can only  be changed due to a direct extension
or a representative summation
after being created.
It is easy to verify that a representative is generated by a bundle
after being created
and that a representative is still generated by a bundle after being extended
given that it is generated by a bundle before the extension.
We then only need to show that $\rep\repsum\rep'$ is still generated by 
a bundle if two representatives $\rep$ and $\rep'$
are generated by bundles.
Figure~\ref{fig:bnd-sum} provides an example 
involving two intervals $[2,10],[3,10]$
whose representatives
are generated by the 
bundles $\Set{\wire{2},\wire{5},\wire{7}},\Set{\wire{3},\wire{7}}$
respectively.
The resulting representative of the summation
is generated by the bundle $\Set{\wire{2},\wire{3},\wire{5}}$ which is the symmetric difference.
In general, 
for two bundles $\bndl$ and $\bndl'$
generating  representatives  $\rep$ and $\rep'$ respectively,
it could happen that the representative $\rep^*$ generated by 
$\bndl\bndlsum\bndl'$ is not equal to $\rep\repsum\rep'$.
However, we have that each cycle in $\rep^*$
is always homologous to the corresponding cycle in $\rep\repsum\rep'$.
The rest of the section formally justifies the claim.

The reader may wonder why we need the boundary module and its representatives at all.
While theoretically $\PersH(\Fcal)$ and the bundles generating the representatives
can be computed independently without considering the boundary module $\Bnd(\Fcal)$,
introducing $\Bnd(\Fcal)$ 
helps us achieve the $O(\filtcnt^2\simpcnt)$ time complexity.
See Remark~\ref{rmk:bnd-mod-nece} in Section~\ref{sec:alg-oview} for a detailed explanation.
 
For any interval in
$\Pers^H(\Fcal_i)\cup\Pers^B(\Fcal_i)$, our algorithm  maintains
a wire bundle generating its 
representative. Proposition~\ref{prop:wiresum} lets us replace representatives with wire bundles.

\begin{proposition}
    Let $[b,i],[b',i]\in \PersH(\Fcal_i)\cup\PersB(\Fcal_i)$ and $b\lessB b'$. Suppose that $W$ and $W'$ generate a representative for $[b,i]$ and
    $[b',i]$ respectively. 
    Then, the sum $W\bndlsum W'$ generates a representative for $[b',i]\in \PersH(\Fcal_i)\cup\PersB(\Fcal_i)$.
    \label{prop:wiresum}
\end{proposition}

\revise{Before proving Proposition~\ref{prop:wiresum}, 
we prove a result (Proposition~\ref{prop:alive-till}) that says wires in a bundle for an interval $I$ added with other intervals produce only boundaries
outside the interval $I$ and those boundaries reside in the
respective complexes.}
This, in turn, helps to prove
Proposition~\ref{prop:wiresum}.

{\blue Each time we extend an interval 
$[b,i-1]$ in $\Pers^H(\Fcal_{i-1})$
(resp.\ $\Pers^B(\Fcal_{i-1})$)  
to $[b,i]$ in $\Pers^H(\Fcal_i)$ (resp.\ $\Pers^B(\Fcal_i)$), 
the birth index $b$ does not change.}
So we denote the bundle associated with $[b,i]$ as $\bndl^b$
in this section.
After being created,
 $W^b$ only changes when another bundle $W^x$ is added to it because
the representative generated by $W^x$ needs to be added to the representative 
for $[b,i]$
generated by $W^b$. 
\begin{definition}
    A boundary bundle $W$ is said to be \defemph{alive until index $b$}
    if the cycle
    $\cyc_\alpha=\sum_{\wire{j}\in W, j\leq \alpha}\wire{j}$ is in $\Bnd(K_\alpha)$ for every $\alpha \leq b$.
    Notice that $\cyc_\aG$ is the empty chain if
    there is no $\wire{j}\in W$ s.t.\ $j\leq \aG$.
\end{definition}

\begin{proposition}\label{prop:alive-till}
    Let $[b,i]\in \Pers^H(\Fcal_i)$ with
    $K_{b-1}\bakinctosp{} K_b$
    being backward, or $[b,i]\in \Pers^B(\Fcal_i)$. 
    Let $\bar{W}^b\subseteq W^b$ be defined as $\bar{W}^b=\{\wire{j}\in W^{b} \,|\, j< b\}$.
    Then,
    $\bar{W}^b
    $ 
    is a boundary
    bundle alive until $b$.
\end{proposition}
\begin{proof}
    Let $X$ be the set containing each index $x\leq i$ so that either 
    $x\in \pinds^H(\Fcal_i)$ with backward $K_{x-1}\bakinctosp{}K_x$, or simply
    $x\in \pinds^B(\Fcal_i)$. Let
    $\bar{W}^x=\{\wire{j}\in W^x\,|\, j < x\}$
    for each $x\in X$.
       
    Let $a_0,a_1,\ldots,a_k$
    denote the series of all operations that change a bundle $W^x$ for $x\in X$,
    i.e., each $a_j$
    either creates a bundle $W^x$
    at an index $x\in X$ or
    sums a bundle $W^y$ to $W^x$ for $x\in X$. 
    Notice that $y$ is necessarily in $X$ because the bundle summation respects the order `$\lessB$'  in Definition~\ref{dfn:ind-order}.
    {\blue We show by induction on the number of operations $k$}
    that the  bundle $W^x$ for any  $x\in X$ maintains the property that the derived $\bar{W}^x$ is a boundary  bundle
    alive until index $x$. The operation $a_0$ starts a representative
    with a single cycle $\cyc\in \Zyc(K_x)$ at some index $x\in X$ with the
    wire $\wire{x}=\cyc$.
    The bundle $W^x$ then equals $\Set{\wire{x}}$ and
    the claim is trivially true.

    For the inductive step, assume that the claim is true after an operation
    $a_\ell$ for $\ell\geq 0$. If the the operation $a_{\ell+1}$ starts a representative,
    the claim holds trivially. Assume that $a_{\ell+1}$ adds a wire bundle
    $W^y$, $y\in X$, \revise{to $W^x$}. 
    By the inductive hypothesis,
    $\bar{W}^x=\{\wire{j} \,|\, \wire{j}\in W^x, j < x\}$ and $\bar{W}^y=\{\wire{j} \,|\, \wire{j}\in W^y, j < y\}$ are boundary bundles alive until $x$ and $y$ respectively. 
    There are two possibilities:
    
    (i) $y> x$: 
    Let $W=W^x\bndlsum W^y$. Observe that, for  $\alpha \leq x$,
    the cycle $z_\alpha=\sum_{\wire{j}\in W, j\leq \alpha}\wire{j}
    = \sum_{\wire{j}\in \bar{W}^x, j\leq \alpha}\wire{j}+\sum_{\wire{j}\in \bar{W}^y, j\leq \alpha}\wire{j}$
    is a boundary in $K_\alpha$ because the
    two cycles given by the two sums on RHS are boundaries 
    in $K_\alpha$. 
    It follows that $\bar{W}=\{\wire{j}\in W \,|\, j< x\}$
    is a boundary
    bundle alive until $x$
    and the inductive hypothesis still holds for $x$.
    
    (ii) $y < x$: Let $W=W^x\bndlsum W^y$. In this case,
    $y\in\PersB(\Fcal)$.
    It can be verified that, since $y\in\PersB(\Fcal)$,
    $W^y$ is necessarily a boundary bundle
    because the bundle summations respect 
    the order in Definition~\ref{dfn:ind-order}.
    Then, the bundle
    $W'=\{\wire{j}\,|\, \wire{j}\in W^y, j < x\}$ is a boundary bundle
    alive until $x$. 
    By the inductive hypothesis,
    the wire bundle $\bar{W}^x$ is a boundary bundle
    alive until $x$. Therefore, the
    sum $W'\bndlsum \bar{W}^x$, which is the updated $\bar{W}^x$, is a boundary bundle 
    alive until $x$; the inductive hypothesis follows.
\end{proof}

\begin{proof}[Proof of Proposition~\ref{prop:wiresum}]
Let $\rep=\{\cyc_\alpha\,|\,\alpha\in [b,i]\}$, 
$\rep'=\{\cyc'_\alpha\,|\,\alpha\in [b',i]\}$ be the representatives generated
by $W$ and $W'$ respectively. We have the following cases to consider:

    \noindent
    {\bf Case 1}, $b<b'$: In this case, every cycle $\bar{\cyc}_\alpha$,
    $\alpha\in[b',i]$, in
    $\rep\bndlsum\rep'$ satisfies that $\bar{\cyc}_\alpha=\cyc_\alpha+\cyc'_\alpha$. Since
    $\cyc_\alpha=\sum_{\wire{j}\in W,j\leq \alpha} \wire{j}$ and
    $\cyc'_\alpha=\sum_{\wire{j}\in W',j\leq \alpha} \wire{j}$, we have that
    $$
    \bar{\cyc}_\alpha=\sum_{\wire{j}\in W,j\leq\alpha} \wire{j}+ \sum_{\wire{j}\in W',j\leq \alpha} \wire{j}
    = \sum_{\wire{j}\in W\bndlsum W', j\leq \alpha} \wire{j}.
    $$ This means that
    $W\bndlsum W'$ generates $\rep\bndlsum\rep'$, a representative for $[b',i]$ by Proposition~\ref{prop:repsum}.

    \noindent
    {\bf Case 2}, $b' < b$: We have 
    $\bar{\cyc}_\alpha$ in $\rep\bndlsum\rep'$ equals $\cyc'_\alpha$
    for $b'\leq \alpha <b$.
    However,
    the wire bundle $W$ may have wires in $\bar{W}=\{\wire{j}\in W\,|\, j < b\}$
    whose addition to $\cyc'_\aG$, $b'\leq\alpha<b$, may create
    a different cycle in the representative generated by $W\bndlsum W'$. By Proposition~\ref{prop:alive-till}, $\bar{W}$ is necessarily a boundary
     bundle alive until  index $b$.
     Let $\bar{\cyc}'_\alpha$ be 
     the cycle at index $\aG$ in the representative generated by $W\bndlsum W'$,
     where $b'\leq \alpha <b$.
     Then, 
   $\bar{\cyc}'_\alpha=\sum_{\wire{j}\in W\bndlsum W', j\leq \alpha}\wire{j}=\cyc'_\alpha+\sum_{\wire{j}\in \bar{W},{j\leq \alpha}}\wire{j}$, 
   which means that $\bar{\cyc}'_\alpha$
    is homologous
    to $\cyc'_\alpha$. Hence, $\bar{\cyc}'_\alpha$ can be taken as a cycle in a
    representative for $[b',i]$. 
    This means that $W\bndlsum W'$ generates a representative for $[b',i]$.
\end{proof}

\begin{theorem}
There is a wire bundle $W=\{w_\iota \,|\, \iota \in \pinds^H(\Fcal)\cup \pinds^B(\Fcal)\}$ so that
a representative for any $[b,d]\in\PersH(\Fcal)\union \PersB(\Fcal)$
is generated by a wire bundle that is a subset of $W$.
\label{thm:bundle}
\end{theorem}

\begin{proof}
{\blue We give a constructive proof
by showing that we have 
a wire bundle $W_i=\{w_\iota\,|\, \iota\in \pinds^H(\Fcal_i)\cup\pinds^B(\Fcal_i)\}$ so that for every 
$[b,d]\in\Pers^H(\Fcal_i)\union\Pers^B(\Fcal_i)$, there is a wire bundle $W^{[b,d]}\subseteq W_i$ that generates a representative for $[b,d]$.
The base case when $i=0$ holds trivially.}
For the inductive step, consider extending the filtration $\Fcal_i$ to $\Fcal_{i+1}$ while assuming the hypothesis for $\Fcal_i$. Since any
$[b,d]\in\Pers^H(\Fcal_i)\union\Pers^B(\Fcal_i)$
where $d< i$ is not affected by
the extension, we do not consider them in the arguments below.

\smallskip
\noindent
{\bf Case 1,} $K_i\inctosp{\sigma_i}K_{i+1}$ and
$i+1\in \pinds^H(\Fcal_{i+1})$: Any $[b,i]\in\Pers^H(\Fcal_i)$ extends to
$[b,i+1]\in\Pers^H(\Fcal_{i+1})$ because the representative cycle $\cyc_i$ at index $i$ for $[b,i]$ is also in $K_{i+1}$ and thus we choose $\cyc_{i+1}=\cyc_i$ for
$[b,i+1]$. Then, the wire bundle $W^{[b,i]}$ also represents $[b,i+1]$.
The same holds for intervals in $\Pers^B(\Fcal_i)$.
We also have a new interval $[i+1,i+1]\in \Pers^H(\Fcal_{i+1})$. 
Let
a new wire $\wire{i+1}$ 
be any cycle in $\Zyc(K_{i+1})$  containing $\sigma_i$.
We have that the bundle $\Set{\wire{i+1}}$
generates a representative for $[i+1,i+1]$.
Subsets of the wire bundle $W_{i+1}=W_i\cup \{\wire{i+1}\}$ then represent intervals  in both $\Pers^H(\Fcal_{i+1})$ and $\Pers^B(\Fcal_{i+1})$. 

\smallskip
\noindent
{\bf Case 2,} $K_i\inctosp{\sigma_i} K_{i+1}$ 
and $i\in \ninds^H(\Fcal_{i})$: In this case, $\partial \sigma_i$
becomes a boundary in $K_{i+1}$,  an interval in $\Pers^H(\Fcal_i)$
does not extend to $i+1$, and a new interval $[i+1,i+1]$ in $\Pers^B(\Fcal_{i+1})$ begins. To determine the interval $[b,i]\in \Pers^H(\Fcal_i)$ that 
does not extend to $i+1$,
consider
the cycle $\partial \sigma_i$ which is in $\Zyc(K_i)\setminus \Bnd(K_i)$.
Let $[b_1,i],\ldots,[b_k,i]$ be all the
intervals in $\Pers^H(\Fcal_i)$ and
$\Pers^B(\Fcal_i)$ 
with representatives $\rep_1,\ldots, \rep_k$ respectively.
Let $\cyc_1,\ldots,\cyc_k$
be their cycles at index $i$ respectively. Since these cycles form
a basis for $\Zyc(K_i)$ by Proposition~\ref{prop:cyclebasis},
the cycle $\partial \sigma_i$ is a linear combination of them. Without loss of generality (WLOG), assume that after reindexing, $\partial \sigma_i=\cyc_1+\cdots+\cyc_\ell$ 
for some $\ell\leq k$
where $b_1\lessB\cdots\lessB b_\ell$.
Add the representatives $\rep_1,\ldots, \rep_{\ell-1}$ to $\rep_\ell$ to obtain a new representative $\rep_\ell'$ for $[b_\ell,i]\in\PersH(\Filt{i})$ (Proposition~\ref{prop:repsum}). The cycle of $\rep_\ell'$ at index $i$ is $\partial\sigma_i$ by construction which becomes a boundary in $K_{i+1}$. Therefore, $\rep'_{\ell}$ is
a representative for $[b_\ell,i]\in\PersH(\Filt{i+1})$
and
$W^{[b_1,i]}\bndlsum\cdots\bndlsum W^{[b_{\ell},i]}$ 
represents $[b_\ell,i]\in\PersH(\Filt{i+1})$
by Proposition~\ref{prop:wiresum}.
All other intervals
in $\Pers^H(K_i)$ and $\Pers^B(K_i)$ extend to $\Pers^H(K_{i+1})$ and
$\Pers^B(K_{i+1})$ with their wire bundles remaining the same. 
A new interval $[i+1,i+1]\in \Pers^B(\Fcal_{i+1})$ 
begins
whose representative 
is given by the cycle $\partial \sigma_i$.
So, the wire $\wire{i+1}=\partial \sigma_i$ 
represents this interval in $\Pers^B(\Fcal_{i+1})$. Subsets of the wire bundle $W_{i+1}=W_i\cup\{ \wire{i+1}\}$ then generate
representatives for
all intervals in $\Pers^H(\Fcal_{i+1})\cup\Pers^B(\Fcal_{i+1})$.

\noindent
{\bf Case 3}, $K_i\bakinctosp{\sigma_i} K_{i+1}$ and
$i+1\in \pinds^H(\Fcal_{i+1})$: In this case, an interval $[b,i]\in \Pers^B(\Fcal_i)$ 
does not extend to $i+1$
and a new interval $[i+1,i+1]\in \Pers^H(\Fcal_{i+1})$ begins. 
 Let $[b_1,i],\ldots,[b_k,i]$ be all the
intervals in $\Pers^B(\Fcal_i)$ 
with representatives $\rep_1,\ldots, \rep_k$ respectively,
and let $z^j_i$
be the cycle at index $i$ in $\rep_j$ for each $1\leq j\leq k$.
Since $z^j_i\in\Bnd(K_i)$, $z^j_i$ has a `bounding chain' $c^j\in\Chn(K_i)$
s.t.\ $z^j_i=\partial(c^j)$.
Assuming after reindexing $z^j_1,\ldots,z^j_\ell$
are all the cycles
whose bounding chains contain $\sigma_i$
where $b_1\lessB\cdots\lessB b_\ell$.
We add $\rep_1$ to $\rep_2,\ldots,\rep_\ell$ to remove $\sigma_i$ from their bounding chains. 
Then, the new representatives $\rep'_2:=\rep_1\bndlsum\rep_2,\ldots,
\rep'_\ell:=\rep_1\bndlsum\rep_\ell$ for the intervals $[b_2,i],\ldots,[b_\ell,i]$ 
can extend to $i+1$ because their bounding chains now do not contain $\sigma_i$.
By Proposition~\ref{prop:wiresum},
$W^{[b_j,i]}\bndlsum W^{[b_1,i]}$ represents
$[b_j,i+1]\in\PersB(\Filt{i+1})$ for $2\leq j\leq \ell$.
So, we update $W^{[b_j,i]}$ as $W^{[b_j,i]}\bndlsum W^{[b_1,i]}$ for $2\leq j\leq k$.
The interval $[b_1,i]$ does not extend to $i+1$
with the wire bundle
$W^{[b_1,i]}$ still representing $[b_1,i]\in\PersB(\Filt{i+1})$.
A new interval $[i+1,i+1]\in \Pers^H(\Fcal_{i+1})$ begins with a representative  $\Set{\partial \sigma_i\textit{}}$ which is generated by a new wire
$\wire{i+1}=\partial\sigma_i$. Subsets of the wire bundle $W_{i+1}=W_i\cup\{\wire{i+1}\}$ then generate representatives for all intervals in $\Pers^H(\Fcal_{i+1})\cup\Pers^B(\Fcal_{i+1})$.

\smallskip
\noindent
{\bf Case 4}, $K_i\bakinctosp{\sigma_i} K_{i+1}$ and
$i\in \ninds^H(\Fcal_{i})$: In this case, an interval $[b,i]\in \Pers^H(\Fcal_i)$
does not extend to $i+1$.
Let $\rep_1,\ldots,\rep_k$ be all the representatives for 
$[b_1,i],\ldots, [b_k,i]\in \Pers^H(\Fcal_i)$ respectively 
whose cycles
at index $i$ contain $\sigma_i$.
WLOG, assume that $b_1\lessB\cdots\lessB b_k$.
We cannot extend these representatives
to $i+1$ because $\sigma_i\nsubseteq K_{i+1}$.
We add $\rep_1$ to $\rep_2,\ldots,\rep_k$
to obtain new representatives $\rep_2',\ldots,\rep_k'$ for the intervals 
whose cycles at index $i$ now do not contain $\sigma_i$. 
Similar to  previous cases, the  bundle $W^{[b_j,i]}\bndlsum W^{[b_1,i]}$ 
represents $[b_j,i+1]\in\PersH(\Filt{i+1})$
for $2\leq j\leq k$. 
So, we update 
$W^{[b_j,i]}$ 
as $W^{[b_j,i]}\bndlsum W^{[b_1,i]}$ for $2\leq j\leq k$. 
The interval $[b_1,i]$ does not extend to $i+1$
and $\rep_1$ remains a representative for $[b_1,i]\in \PersH(\Filt{i+1})$. 

\smallskip
To finish the proof, we also need to show that the zigzag barcodes
we have are correct whenever we proceed from $\Filt{i}$
to $\Filt{i+1}$.
Since all intervals we have admit representatives,
the correctness of the barcodes follows from Proposition~\ref{prop:pn-paring-w-rep}
presented in Section~\ref{sec:alg-oview}.
\end{proof}



\section{Algorithm}
\label{sec:alg-oview}
We present the $O(\filtcnt^2\simpcnt)$ algorithm \textsc{WiredZigzag} computing $\PersH(\Fcal)$,
$\PersB(\Fcal)$, and their representatives based on exposition in the previous section.
As mentioned,
the general idea of the algorithm is to 
maintain a wire bundle for each interval in $\PersH(\Fcal)$
and $\PersB(\Fcal)$
so that the bundle generates a representative for the interval.
In each iteration $i$,
the algorithm processes the inclusion $K_i\leftrightarrowsp{\fsimp_{i}}K_{i+1}$
in $\Fcal$ starting with $i=0$.
Before iteration $i$,
we assume that
we have computed
all intervals in $\PersH(\Filt{i})$ 
and $\PersB(\Filt{i})$ 
along with the wire bundles.
The aim of iteration $i$ is  to compute those for 
$\PersH(\Filt{i+1})$ 
and $\PersB(\Filt{i+1})$.
In each iteration $i$,
we  have two sets of \emph{active} intervals 
(ending with $i$)
for $\PersH(\Filt{i})$ 
and $\PersB(\Filt{i})$ respectively,
\[
\bigSet{[\hbirth_j,i]\in\PersH(\Filt{i})\mid j=1,2,\ldots,r},\quad
\bigSet{[\bbirth_k,i]\in\PersB(\Filt{i})\mid k=1,2,\ldots,q}\]
where $r$ is the dimension of $\Hm(K_i)$
and $q$ is the dimension of $\Bnd(K_i)$.
All non-active intervals in $\PersH(\Filt{i})$
(resp.\ $\PersB(\Filt{i})$)
are automatically  carried into $\PersH(\Filt{i+1})$
(resp.\ $\PersB(\Filt{i})$)
and their wire bundles  do not change.
For each homology interval $[\hbirth_j,i]$, 
we let
$\bndl^j
$
denote the 
(non-boundary)
bundle 
maintained for $[\hbirth_j,i]$,
and for each boundary interval $[\bbirth_k,i]$, 
we let
$\bndbndl^k
$
denote the 
(boundary)
bundle 
maintained for $[\bbirth_k,i]$.
At the end of the algorithm,
we have
all intervals and bundles in $\PersH(\Filt{\filtcnt})=\PersH(\Fcal)$
and $\PersB(\Filt{\filtcnt})=\PersB(\Fcal)$.
We then generate a representative
for each interval
from its bundle.

\subsection{Maintenance of pivoted matrices}
\label{sec:matrices}
For the computation,
we maintain 
three 0-1 matrices $\cycmat$, $\bndmat$, and $\chnmat$ where each column
represents a chain such that
the $k$-th entry of the column
equals 1 iff the simplex with index $k$ belongs  to the chain.
We also do not differentiate a matrix column and the chain it represents
when describing the algorithm.
In each iteration $i$, 
the following invariants hold:
\begin{enumerate}
    \item $\cycmat$ has $r$ columns each corresponding to
    an active interval in $\PersH(\Filt{i})$ s.t.\ a column $\cmatcol{j}$ 
    equals the last cycle (at index $i$) in the representative for $[\hbirth_j,i]$ generated
    by $\bndl^j$.

    \item $\bndmat$ has $q$ columns each corresponding to
    an active interval in $\PersB(\Filt{i})$ s.t.\ a column 
    $\bndmat[k]$
    equals the last cycle 
    in the representative for $[\bbirth_k,i]$ generated
    by $\bndbndl^k$.

    \item $\chnmat$ also has $q$ columns s.t.\ $\bndmat[k]=\partial(\chnmat[k])$
    for each $k$.
\end{enumerate}

By Proposition~\ref{prop:cyclebasis},
columns in $\cycmat$ and $\bndmat$
form a basis of $\Zyc(K_i)$.
Throughout the algorithm,
we also always ensure that 
columns in $\cycmat$, $\bndmat$, and $\chnmat$ form a basis for 
$\Chn(K_i)$
in each iteration.
This can be inductively proved based on the details
of the algorithm presented in this section.
The detailed justification is omitted.
Let
the \emph{pivot} of 
a matrix column 
be the index
of its lowest entry equal to 1.
Our algorithm  maintains the invariant that \emph{columns 
in $\cycmat$ and $\bndmat$ altogether have distinct pivots}
so that getting the coordinates of a cycle in $\Zyc(K_i)$
in terms of the basis represented by
columns of $\cycmat$ and $\bndmat$
takes $O(\simpcnt^2)$ time.
This is an essential part of our algorithm as demonstrated in its description below.

\begin{remark}
Since a bundle is just a set of wire indices
and we have no more than one wire born at each index,
we maintain all wires in a matrix representing the containments
$\wire{i}\mapsto \{\sigma\in K_i \mid \sigma\in \wire{i}\}$. Similarly, bundles are maintained in 
a matrix representing the containments $W\mapsto \{\wire{i}\in W\}$. 
\end{remark}

\subsection{Iteration $i$ of {\sc WiredZigzag}}
Iteration $i$ of the algorithm has the following 
processes
in different cases
(more details such as
how we ensure the distinct pivots in $\cycmat,\bndmat$
and how to determine the injective/surjective cases
are provided in  Section~\ref{sec:impl}):
\begin{description}
\smallskip
\item[$K_{i}\inctosp{\fsimp_i} K_{i+1}$ is forward, $\linmap_i$ is injective:]

We have:
\begin{description}
\smallskip
    \item[Birth in homology module \textup{(}$i+1\in\pindsH(\Fcal)$\textup{):}]
An interval $[i+1,i+1]\in\PersH(\Filt{i+1})$ active in the next iteration is created.
We find a 
new non-boundary wire $\wire{i+1}$
which is a cycle in $K_{i+1}$
containing $\fsimp_i$
so that condition (i) in  Definition~\ref{dfn:wire} is satisfied.
We also have a new non-boundary
bundle $\Set{\wire{i+1}}$ 
for $[i+1,i+1]\in\PersH(\Filt{i+1})$.
(The validity of the new bundle 
$\Set{\wire{i+1}}$ can be seen by examining 
Definitions~\ref{dfn:rep-cls} and~\ref{dfn:wire}.)
Each $[\hbirth_j,i]\in\PersH(\Filt{i})$ 
extends to be
an active interval
$[\hbirth_j,i+1]\in\PersH(\Filt{i+1})$.
Since $\cmatcol{j}\subseteq K_{i+1}$ because $K_{i}\inctosp{\fsimp_i} K_{i+1}$ is forward, 
$\bndl^j$ stays the same for the next iteration.
Finally, since $\wire{i+1}$ is the cycle at index $i+1$
in the representative generated by the bundle 
$\Set{\wire{i+1}}$,
we add $\wire{i+1}$ as a new column to $\cycmat$
corresponding to the new active interval.
\end{description}
\smallskip
Since $\Bnd(K_i)=\Bnd(K_{i+1})$,
each $[\bbirth_k,i]\in\PersB(\Filt{i})$ 
extends to be
an active interval
$[\bbirth_k,i+1]\in\PersB(\Filt{i+1})$
and the bundle $\bndbndl^k$
 stays the same.


\smallskip
\item[$K_{i}\inctosp{\fsimp_i} K_{i+1}$ is forward, $\linmap_i$  is surjective:]
Both of the following happen:
\begin{description}
\smallskip
\item[Death  in homology module \textup{(}$i\in\nindsH(\Fcal)$\textup{):}]
By performing reductions on $\partial\fsimp_i$
and the columns in $\cycmat$ and $\bndmat$,
we find a subset of columns ($J\subseteq \Set{1,\ldots,r}$) in $\cycmat$
such that
\begin{equation}\label{eqn:sum-to-zero}
\sum_{j\in J}[\cmatcol{j}]=[\partial\fsimp_i]\text{ in }\Hm(K_i).
\end{equation}
Let $\hbirth_\lG$ be the maximum birth index in $\Set{\hbirth_j\mid j\in J}$
w.r.t the order `$\lessB$'.
We have that
$[\hbirth_\lG,i]$ ceases to be active,
i.e., $[\hbirth_\lG,i]\in\PersH(\Filt{i+1})$.
Let $\bndl^*$ be the sum of all the bundles in
$\Set{\bndl^j\mid j\in J}$.
We have that $\bndl^*$ generates a representative $\rep^*$
for $[\hbirth_\lG,i]\in\PersH(\Filt{i+1})$.
(To see this, notice that the last cycle in
$\rep^*$
is $\sum_{j\in J}\cmatcol{j}$
which is homologous to $\partial\fsimp_i$ in $K_i$,
and so the  death condition  in Definition~\ref{dfn:rep-cls}
is satisfied.
The validity of $\bndl^*$ then follows from Proposition~\ref{prop:wiresum}.)
For each $j\in\Set{1,\ldots,r}\setminus\Set{\lG}$,
$[\hbirth_j,i]\in\PersH(\Filt{i})$ 
extends to be
$[\hbirth_j,i+1]\in\PersH(\Filt{i+1})$
for which $\bndl^j$ stays the same
because $K_{i}\inctosp{\fsimp_i} K_{i+1}$ is forward.
Finally, we delete $\cmatcol{\lG}$ from $\cycmat$.

\smallskip
\item[Birth in boundary module \textup{(}$i+1\in\pindsB(\Fcal)$\textup{):}]
A new active interval $[i+1,i+1]\in\PersB(\Filt{i+1})$ is created.
We have a 
new boundary wire $\wire{i+1}=\partial\fsimp_i$
satisfying condition (iii) in  Definition~\ref{dfn:wire}.
We also have a new boundary
bundle $\Set{\wire{i+1}}$ 
for $[i+1,i+1]\in\PersB(\Filt{i+1})$.
Each $[\bbirth_k,i]\in\PersB(\Filt{i})$ 
extends to be
$[\bbirth_k,i+1]\in\PersB(\Filt{i+1})$
for which $\bndbndl^k$ stays the same.
Finally,
we add  $\partial\fsimp_i$ as a new column to $\bndmat$
and
add a new column containing only $\fsimp_i$ to $\chnmat$.
\end{description}

\smallskip
\item[$K_{i}\bakinctosp{\fsimp_i} K_{i+1}$ is backward, $\linmap_i$  is surjective:]
Both of the following happen:
\begin{description}
\smallskip
    \item[Birth in homology module \textup{(}$i+1\in\pindsH(\Fcal)$\textup{):}]
A new active interval $[i+1,i+1]\in\PersH(\Filt{i+1})$ is created.
We find a 
new non-boundary wire $\wire{i+1}$
which is
a cycle homologous to $\partial\fsimp_i$ in $K_{i+1}$
so that condition (ii) in  Definition~\ref{dfn:wire} is satisfied.
The rest of the processing is the same as in
the previous birth event for the homology module.
Notice that each
$\bndl^j$ stays the same
because $\Zyc(K_i)=\Zyc(K_{i+1})$.

\smallskip
\item[Death  in boundary module \textup{(}$i\in\nindsB(\Fcal)$\textup{):}]
{\blue Since $\fsimp_i$ is not in a cycle in $K_i$}
and  columns in $\cycmat$, $\bndmat$, and $\chnmat$
form a basis of $\Chn(K_i)$,
at least one column 
in $\chnmat$
contains $\fsimp_i$.
Whenever there are 
two columns $\chnmat[j]$, $\chnmat[k]$ in $\chnmat$
containing $\fsimp_i$
with $\bbirth_k\lessB\bbirth_j$,
set
$\chnmat[j]=\chnmat[j]+\chnmat[k]$,
$\bndmat[j]=\bndmat[j]+\bndmat[k]$,
and $\bndbndl^j=\bndbndl^j\bndlsum\bndbndl^k$
to remove $\fsimp_i$ from $\chnmat[j]$.
After this, only one column $\chnmat[\lG]$
 in $\chnmat$ contains $\fsimp_i$
and we have that $[\bbirth_\lG,i]\in\PersB(\Filt{i+1})$ ceases to be active.
Notice that  $\bndbndl^\lG$ still generates a
representative for $[\bbirth_\lG,i]\in\PersB(\Filt{i+1})$.
For each $k\in\Set{1,\ldots,q}\setminus\Set{\lG}$,
$[\bbirth_k,i]\in\PersB(\Filt{i})$ 
extends to be
$[\bbirth_k,i+1]\in\PersB(\Filt{i+1})$
for which $\bndbndl^k$ now stays the same because
$\fsimp_i\not\in\chnmat[k]$
so that $\bndmat[k]\in\Bnd(K_{i+1})$.
Finally, we delete $\bndmat[\lG]$ from $\bndmat$
and delete $\chnmat[\lG]$ from $\chnmat$.
\end{description}

\smallskip
\item[$K_{i}\bakinctosp{\fsimp_i} K_{i+1}$ is backward, $\linmap_i$  is injective:]
We have:
\begin{description}
\smallskip
\item[Death  in homology module \textup{(}$i\in\nindsH(\Fcal)$\textup{):}]
We have that at least one column 
in $\cycmat$
contains $\fsimp_i$.
(To see this, notice that
$\fsimp_i$ cannot be in a column in $\bndmat$
because $\fsimp_i$ has no cofaces in $K_i$.
So $\fsimp_i$ has to be in a column in $\cycmat$
because  $\cycmat$ and $\bndmat$ provide a basis for $\Zyc(K_i)$
and there is a cycle in $K_i$ containing $\fsimp_i$.)
Whenever there are two columns $\cycmat[j]$, $\cycmat[k]$
in $\cycmat$
with $\hbirth_k\lessB\hbirth_j$
containing $\fsimp_i$,
set
$\cycmat[j]=\cycmat[j]+\cycmat[k]$
and $\bndl^j=\bndl^j\bndlsum\bndl^k$
to remove $\fsimp_i$ from $\cycmat[j]$.
After this, only one column $\cmatcol{\lG}$ in $\cycmat$
contains  $\fsimp_i$ and
we have that $[\hbirth_\lG,i]\in\PersH(\Filt{i+1})$ ceases to be active.
The remaining processing resembles 
what is done
in the
death event for the boundary module
and is omitted.
Notice that we also need to remove $\fsimp_i$ from $\chnmat$
and the details are provided in Section~\ref{sec:impl}.
\end{description}
\smallskip
Since $\Bnd(K_i)=\Bnd(K_{i+1})$,
each $[\bbirth_k,i]\in\PersB(\Filt{i})$ 
extends to be
$[\bbirth_k,i+1]\in\PersB(\Filt{i+1})$
and the bundle $\bndbndl^k$
 stays the same.
\end{description}

\begin{remark}\label{rmk:pairing-struct}
We can also consider 
our algorithm 
to have a `pairing of birth/death points'
structure as adopted by the algorithm
for computing standard persistence~\cite{edelsbrunner2000topological},
where, e.g.,  $\hbirth_1,\ldots,\hbirth_r$ 
are carried as `unpaired' birth indices to be paired
for the homology module.
\end{remark}

The following proposition from~\cite[Proposition~9]{dey2021computing}
helps draw our conclusion:

\begin{proposition}
\label{prop:pn-paring-w-rep}
Let 
$\pi:\pindsH(\Fcal_i)\to\nindsH(\Fcal_i)$ be a bijection.
If every $b\in\pindsH(\Fcal_i)$ satisfies that $b\leq\pi(b)$ and the interval $[b,\pi(b)]$
has a representative,
then $\PersH(\Fcal_i)=\Set{[b,\pi(b)]\mid b\in\pindsH(\Fcal_i)}$.
\end{proposition}
\begin{remark}
Similar facts hold for $\pindsB(\Filt{i})$, $\nindsB(\Filt{i})$, and $\PersB(\Filt{i})$.
\end{remark}

\begin{theorem}
The barcodes $\PersH(\Fcal)$ and 
$\PersB(\Fcal)$ along with the representatives for the intervals
can be computed in $O(\filtcnt^2\simpcnt)$ time and $O(mn)$ space.
\end{theorem}
\begin{proof}
First, to see that the algorithm presented above
runs in  $O(\filtcnt^2\simpcnt)$ time,
we notice that there are no more than $O(\simpcnt)$ 
summations of matrix columns and wire bundles
in each iteration,
which can be verified from the details presented in this 
section and Section~\ref{sec:impl}.
Hence, each iteration runs in $O(\filtcnt\simpcnt)$ time
where the costliest steps are the bundle summations.
At the end of the algorithm, we also need
to generate a representative for each interval from the maintained bundle.
Generating representatives 
for all the $O(\filtcnt)$ intervals
can be done in 
$O(\filtcnt^2\simpcnt)$ time (see the Algorithm \textsc{ExtRep}).
The $O(\filtcnt^2\simpcnt)$ complexity then follows. The space complexity follows from maintaining $O(m)$ wires each being a cycle of size $O(n)$,
$O(n)$ bundles for the active intervals each of size $O(m)$,
and the three matrices of size at most $O(n^2)$.

Based on Proposition~\ref{prop:pn-paring-w-rep},
the correctness of the algorithm follows from
the fact that wire bundles 
always correctly generate representatives for the intervals 
in our algorithm.
The validity of the wire bundles 
follows from Proposition~\ref{prop:wiresum}
(the only way a bundle changes after being created is by summations)
and how we assign a bundle to an interval
in the algorithm
when an interval is created or ceases to be active (finalized).
\end{proof}

\begin{remark}\label{rmk:bnd-mod-nece}
The key to achieving the $O(\filtcnt^2\simpcnt)$ time complexity are the following two invariants maintained in our 
algorithm as described in Section~\ref{sec:matrices}:
(i) pivots for the matrices 
$\cycmat$ and $\bndmat$ are always distinct
and (ii) $\cmatcol{j}$ always
    equals the last cycle in the representative for $[\hbirth_j,i]$ generated
    by $\bndl^j$.
By invariant (i), we can obtain the sum in Equation~(\ref{eqn:sum-to-zero})
in $O(\simpcnt^2)$ time
by reductions.
By invariant (ii),
we can take the sum $\bndl^*$ of the bundles
$\Set{\bndl^j\mid j\in J}$ 
based on Equation~(\ref{eqn:sum-to-zero})
for the finalized interval $[\hbirth_\lG,i]$
when a death happens in the homology module.
It ensures that
the last cycle in the representative for $[\hbirth_\lG,i]\in\PersH(\Filt{i+1})$
generated by $\bndl^*$ satisfies the death condition in Definition~\ref{dfn:rep-cls}.
As evident in Section~\ref{sec:impl},
in order to maintain the distinctness of pivots, 
one cannot avoid summations of columns in $\bndmat$
to columns in $\cycmat$.
Without incorporating the module $\Bnd(\Fcal)$ 
and its bundles, invariant (ii) would not hold
when columns in $\bndmat$
are summed to columns in $\cycmat$.
\end{remark}

\subsection{Algorithm details}
\label{sec:impl} 

We provide full details for the algorithm presented above.
For a column $c$ of the matrices maintained,
we denote the pivot of $c$ as $\pivot(c)$.
Also,
in our algorithm,
each simplex $\fsimp_i$ added in $\Fcal$ is assigned an 
{\tt ID} $i$.
This means that
a simplex has 
a new {\tt ID} when it is added again 
after being deleted.
We then present the details for the different cases.

\subsubsection{Forward  $K_i\inctosp{\fsimp_{i}}K_{i+1}$}
\label{sec:detail-fwd}
We need to determine whether $\partial\fsimp_{i}$
is already a boundary in $K_i$.
If this is true, 
a new cycle containing $\fsimp_i$ is created in $K_{i+1}$
and $\linmap_i$ is injective;
otherwise, the homology class $[\partial\fsimp_{i}]$
becomes trivial in $\Hm(K_{i+1})$
and $\linmap_i$ is surjective.
To determine this,
we  perform reductions on $\partial\fsimp_i$
and the columns in $\cycmat$ and $\bndmat$
to get a sum 
\begin{equation}\label{eqn:fwd-red}
\partial\fsimp_i=\sum_{j\in J} \cmatcol{j}+\sum_{k\in I}\bndmat[k].
\end{equation}
We then have that $\partial\fsimp_{i}$
is a boundary in $K_i$ iff $J=\emptyset$.

\paragraph{$\linmap_i$ is injective.}
\label{sec:fwd-birth-detail}
Since $\partial\fsimp_i=\sum_{k\in I}\bndmat[k]$,
we let the new wire 
$\wire{i+1}$ 
containing $\fsimp_i$
 be $\wire{i+1}=\fsimp_i+\sum_{k\in I}\chnmat[k]$,
where $\partial(\fsimp_i+\sum_{k\in I}\chnmat[k])=\partial\fsimp_i+\sum_{k\in I}\bndmat[k]=0$.
Notice that as mentioned,
we need to add $\wire{i+1}$ 
as a new column to the matrix $\cycmat$.
Since $\pivot(\wire{i+1})=i$, columns in $\cycmat$
and $\bndmat$  still have distinct pivots.

\paragraph{$\linmap_i$ is surjective.}
\label{sec:fwd-death-detail}

The subset $J$ derived from the reductions
as in \Cref{eqn:fwd-red}
is the same as the subset $J$ 
in Equation~(\ref{eqn:sum-to-zero})
in the corresponding case 
of Section~\ref{sec:alg-oview}.
So the processing for the corresponding case
described 
in Section~\ref{sec:alg-oview}
can be directly performed.
Notice that we add a new column to $\bndmat$ in this case.
Since the pivot of the new column of $\bndmat$
may conflict with the pivot of another column in $\cycmat$ or $\bndmat$,
we use a loop to repeatedly sum two columns 
whose pivots are the same 
until the pivots become distinct again.
In each iteration of the loop, three cases can happen:
\begin{enumerate}
    \item Two columns $\bndmat[j]$ and $\bndmat[k]$ have the same pivot:
    WLOG,
    assume that $\bbirth_k\lessB\bbirth_j$.
    Let $\bndmat[j]=\bndmat[j]+\bndmat[k]$,
     $\chnmat[j]=\chnmat[j]+\chnmat[k]$,
     and $\bndbndl^j=\bndbndl^j\bndlsum\bndbndl^k$.

    \item Two columns $\cmatcol{j}$ and $\bndmat[k]$ have the same pivot:
    We have $\bbirth_{k}\lessB\hbirth_j$.
    Let $\cmatcol{j}=\cmatcol{j}+\bndmat[k]$
    and  $\bndl^j=\bndl^j\bndlsum\bndbndl^k$.

    \item Two columns $\cmatcol{j}$ and $\cmatcol{k}$ have the same pivot:
    WLOG,
    assume that $\hbirth_k\lessB\hbirth_j$.
    Let $\cmatcol{j}=\cmatcol{j}+\cmatcol{k}$
    and $\bndl^j=\bndl^j\bndlsum\bndl^k$.
\end{enumerate}

Since in each iteration of the above loop we change only one column
of $\cycmat$ and $\bndmat$,
there are at most two columns of $\cycmat$ and $\bndmat$ 
with the same pivot at any time.
Hence, the above loop ends in no more than $\simpcnt$ iterations
because the pivot of the two clashed columns is always decreasing.

\subsubsection{Backward  $K_i\bakinctosp{\fsimp_{i}}K_{i+1}$}
\label{sec:detail-bak}
We need to determine whether $\fsimp_{i}$
is in a cycle $\cyc$ in $K_i$.
If this is true, 
$z$ is a cycle in $K_{i}$ but not in $K_{i+1}$
indicating that $\linmap_i$ is injective;
otherwise,
$\linmap_i$ is surjective.
Since columns in $\cycmat$ and $\bndmat$ form a basis for 
$\Zyc(K_i)$,
we only need to check whether 
$\fsimp_{i}$
is in a column in $\cycmat$ or $\bndmat$.
 Moreover, since $\fsimp_i$ has no cofaces in $K_i$,
we have that $\fsimp_i$ cannot be in a boundary in $K_i$.
Therefore, we only need to check 
whether 
$\fsimp_{i}$
is in a column in $\cycmat$.

\paragraph{$\linmap_i$ is surjective.}
\label{sec:bak-birth-detail}


Since columns in $\cycmat$, $\bndmat$, and $\chnmat$ form a basis for $\Chn(K_i)$
and 
$\fsimp_{i}$
is not in a column in $\cycmat$ or $\bndmat$,
we have that $\fsimp_{i}$ must be in at least one column of $\chnmat$.
Since $\fsimp_i\not\in K_{i+1}$,
we need to remove $\fsimp_{i}$ from $\chnmat$
when proceeding from $K_i$ to $K_{i+1}$.
To do this, 
we use a loop to repeatedly sum two columns 
in $\chnmat$ containing $\fsimp_i$
until only one column in $\chnmat$
contains $\fsimp_i$.
Notice that
whenever we sum two columns 
in $\chnmat$, we also need to sum the corresponding columns
in $\bndmat$
and their wire bundles.
Hence, the summations have to respect the order `$\lessB$'.
We use the following loop to perform the summations:

\newlength\ParaAlgIndent
\setlength{\ParaAlgIndent}{1.2em}
        \begin{enumerate}[parsep=0pt,itemsep=0pt]
            \item $\aG_1,\ldots,\aG_\ell\leftarrow$ indices of all columns of $\chnmat$ containing $\fsimp_i$
            \item sort and rename $\aG_1,\ldots,\aG_\ell$ s.t. 
            $\bbirth_{\aG_1}\lessB\cdots\lessB\bbirth_{\aG_\ell}$.
            \item $c_1\leftarrow\matcol{{\chnmat}}{\aG_1}$
            \item $c_2\leftarrow\matcol{{\bndmat}}{\aG_1}$
            \item $U\leftarrow\bndbndl^{\aG_1}$
            \item {\bf for} $\aG\leftarrow\aG_2,\ldots,\aG_\ell$ {\bf do}:
            \item \hspace{\ParaAlgIndent}{\bf if}
            $\pivot(\matcol{{\bndmat}}{\aG})>\pivot(c_2)$ 
            {\bf then}:
            \item \hspace{\ParaAlgIndent}\hspace{\ParaAlgIndent}
            $\matcol{\chnmat}{\aG}\leftarrow\matcol{\chnmat}{\aG}+c_1$ 
            \item \hspace{\ParaAlgIndent}\hspace{\ParaAlgIndent}
            $\matcol{\bndmat}{\aG}\leftarrow\matcol{\bndmat}{\aG}+c_2$ 
            \item \hspace{\ParaAlgIndent}\hspace{\ParaAlgIndent}
            $\bndbndl^{\aG}\leftarrow\bndbndl^{\aG}\bndlsum U$ 
            \item \hspace{\ParaAlgIndent}{\bf else}:
            \item \hspace{\ParaAlgIndent}\hspace{\ParaAlgIndent}
            $\texttt{temp\_c1}\leftarrow\matcol{\chnmat}{\aG}$
            \item \hspace{\ParaAlgIndent}\hspace{\ParaAlgIndent}
            $\matcol{\chnmat}{\aG}\leftarrow\matcol{\chnmat}{\aG}+c_1$ 
            \item \hspace{\ParaAlgIndent}\hspace{\ParaAlgIndent}
            $c_1\leftarrow\texttt{temp\_c1}$
            \smallskip
            \item \hspace{\ParaAlgIndent}\hspace{\ParaAlgIndent}
            $\texttt{temp\_c2}\leftarrow\matcol{\bndmat}{\aG}$
            \item \hspace{\ParaAlgIndent}\hspace{\ParaAlgIndent}
            $\matcol{\bndmat}{\aG}\leftarrow\matcol{\bndmat}{\aG}+c_2$ 
            \item \hspace{\ParaAlgIndent}\hspace{\ParaAlgIndent}
            $c_2\leftarrow\texttt{temp\_c2}$
            \smallskip
            \item \hspace{\ParaAlgIndent}\hspace{\ParaAlgIndent}
            $\texttt{temp\_U}\leftarrow\bndbndl^{\aG}$
            \item \hspace{\ParaAlgIndent}\hspace{\ParaAlgIndent}
            $\bndbndl^{\aG}\leftarrow\bndbndl^{\aG}\bndlsum U$ 
            \item \hspace{\ParaAlgIndent}\hspace{\ParaAlgIndent}
            $U\leftarrow\texttt{temp\_U}$
        \end{enumerate}

We always maintain the following invariants
for the loop:
(i) $c_2=\partial(c_1)$;
(ii) $c_2$ is the last cycle (at index $i$) in the
representative generated by $U$;
(ii) 
the birth index corresponding to $c_2$
(and $U$)
is always less than 
$\bbirth_\aG$ 
in the total order `$\lessB$';
(iv)
$c_2$ along with $\bndmat[{\aG_2}],\ldots,\bndmat[{\aG_\ell}]$
have distinct pivots.
 When the loop terminates,
we are left with a single column $\chnmat[\lG]:=\chnmat[\aG_1]$ in $\chnmat$
containing $\fsimp_i$.
Notice that 
$\bndmat[\lG]=\partial(\chnmat[\lG])=\partial(\chnmat[\lG]\setminus\Set{\fsimp_i})+\partial\fsimp_i$,
where $\chnmat[\lG]\setminus\Set{\fsimp_i}\subseteq K_{i+1}$.
This indicates that 
$\bndmat[\lG]$ is homologous to $\partial\fsimp_i$
in $K_{i+1}$.
So
we let the new wire $\wire{i+1}$ be $\bndmat[\lG]$
and  need to add $\wire{i+1}$ 
as a new column to $\cycmat$.
Notice that  we also delete $\bndmat[\lG]$ and $\chnmat[\lG]$
from $\bndmat$ and $\chnmat$ respectively.
Since the pivot of the newly added column in $\cycmat$
may clash  with that of another column in $\bndmat$ or $\cycmat$,
we need to perform 
summations as in Section~\ref{sec:fwd-death-detail}
to make the pivots distinct again.
Notice that assumptions on the matrices $\cycmat$, $\bndmat$, and $\chnmat$
still
hold. For example, 
columns in $\bndmat$ still form a basis for $\Bnd(K_{i+1})$
because columns in $\bndmat$ are still linearly independent
and the dimension of $\Bnd(K_{i+1})$ is one less than that of $\Bnd(K_{i})$.

\paragraph{$\linmap_i$ is injective.}
\label{sec:bak-death-detail}

We first update $\chnmat$ so that no columns of $\chnmat$ contain $\fsimp_i$.
            Let $\matcol{\cycmat}{k}$ be a column of $\cycmat$ containing $\fsimp_i$.
            For each column
             $\matcol{\chnmat}{j}$ of $\chnmat$ containing $\fsimp_i$,
            set $\matcol{\chnmat}{j}=\matcol{\chnmat}{j}+\matcol{\cycmat}{k}$.
            Notice that $\partial(\matcol{\chnmat}{j})$ stays the same
            but the updated $\matcol{\chnmat}{j}$ does not contain $\fsimp_i$. 

As indicated in Section~\ref{sec:alg-oview},
whenever there are 
two columns in $\cycmat$
which contain $\fsimp_i$,
we sum the 
two columns and their corresponding bundles
to remove $\fsimp_i$ from one column.
We implement the summations as follows,
which is similar to the loop in Section~\ref{sec:bak-birth-detail}:
        \begin{enumerate}[parsep=0pt,itemsep=0pt]
            \item $\aG_1,\ldots,\aG_\ell\leftarrow$ indices of all columns of $\cycmat$ containing $\fsimp_i$
            \item sort and rename $\aG_1,\ldots,\aG_\ell$ s.t. 
            $\hbirth_{\aG_1}\lessB\cdots\lessB\hbirth_{\aG_\ell}$.
            \item $z\leftarrow\matcol{{\cycmat}}{\aG_1}$
            \item $W\leftarrow\bndl^{\aG_1}$
            \item {\bf for} $\aG\leftarrow\aG_2,\ldots,\aG_\ell$ {\bf do}:
            \item \hspace{\ParaAlgIndent}{\bf if}
            $\pivot(\matcol{{\cycmat}}{\aG})>\pivot(z)$ 
            {\bf then}:
            \item \hspace{\ParaAlgIndent}\hspace{\ParaAlgIndent}
            $\matcol{\cycmat}{\aG}\leftarrow\matcol{\cycmat}{\aG}+z$ 
            \item \hspace{\ParaAlgIndent}\hspace{\ParaAlgIndent}
            $\bndl^{\aG}\leftarrow\bndl^{\aG}\bndlsum W$ 
            \item \hspace{\ParaAlgIndent}{\bf else}:
            \item \hspace{\ParaAlgIndent}\hspace{\ParaAlgIndent}
            $\texttt{temp\_z}\leftarrow\matcol{\cycmat}{\aG}$
            \item \hspace{\ParaAlgIndent}\hspace{\ParaAlgIndent}
            $\matcol{\cycmat}{\aG}\leftarrow\matcol{\cycmat}{\aG}+z$ 
            \item \hspace{\ParaAlgIndent}\hspace{\ParaAlgIndent}
            $z\leftarrow\texttt{temp\_z}$
            \smallskip
            \item \hspace{\ParaAlgIndent}\hspace{\ParaAlgIndent}
            $\texttt{temp\_W}\leftarrow\bndl^{\aG}$
            \item \hspace{\ParaAlgIndent}\hspace{\ParaAlgIndent}
            $\bndl^{\aG}\leftarrow\bndl^{\aG}\bndlsum W$ 
            \item \hspace{\ParaAlgIndent}\hspace{\ParaAlgIndent}
            $W\leftarrow\texttt{temp\_W}$
            \item delete the column $\matcol{{\cycmat}}{\aG_1}$ 
            from $\cycmat$
        \end{enumerate}

In the above pseudocodes, 
$\aG_1$ is the index `$\lG$' as in the corresponding case in Section~\ref{sec:alg-oview}.

\subsubsection{Time complexity}
From the algorithm details provided,
it could be verified that each iteration takes no more than
$O(mn)$ time. For example, the reductions to get the formula in \Cref{eqn:fwd-red}
take no more than $O(n^2)$ time because 
columns in $\cycmat$ and $\bndmat$ have distinct pivots.

\newcommand{\dec}[1]{\tilde{{#1}}}
\subsection{Adaption to non-simplex-wise filtrations}

So far, we have assumed that the input to our algorithm is
a simplex-wise zigzag filtration; we briefly describe here
how to obtain representatives using the algorithm when the input
is not simplex-wise. Assume that we have the following non-simplex-wise
zigzag filtration
\begin{equation*}
\dec{\filt}: \dec{K}_0 \leftrightarrow \dec{K}_1 \leftrightarrow 
\cdots \leftrightarrow \dec{K}_\filtlen.
\end{equation*}
We can build a simplex-wise zigzag filtration $\filt$ as follows
(notice that below
we consider
$\filt$ as a sequence of insertions or deletions of simplices):
\begin{enumerate}
    \item 
    Insert all the simplices in $\dec{K}_0$ such that 
    a simplex is always added after its proper faces.
    (This can be done by, e.g., inserting the simplices based on increasing
    order of dimension.)
    \item For $i=0,1,\ldots,\filtlen-1$:
    \begin{itemize}
        \item If $\dec{K}_{i}\incto\dec{K}_{i+1}$ is forward: 
        insert all simplices in $\dec{K}_{i+1}\setminus\dec{K}_{i}$
        such that
        a simplex is always added after its proper faces.
        \item If $\dec{K}_{i}\bakincto\dec{K}_{i+1}$ is backward: 
        delete all simplices in $\dec{K}_{i}\setminus\dec{K}_{i+1}$
        such that
        a simplex is always deleted before its proper faces.
        (This can be done by, e.g., deleting the simplices based on decreasing
        order of dimension.)
    \end{itemize}
\end{enumerate}

Let $\filt$ be of the following form
\begin{equation*}
\filt: \emptyset=K_0 \leftrightarrow K_1 \leftrightarrow
\cdots \leftrightarrow K_\filtcnt.
\end{equation*}
Notice that each complex $\dec{K}_{\iG}$ in $\dec{\filt}$ is equal to
some complex $K_i$ in $\filt$.
Let $\indmap:\{0,1,\ldots,\filtlen\}\to\{0,1,\ldots,\filtcnt\}$ be
an index map  
such that $\dec{K}_{\iG}=K_{\indmap(\iG)}$
for each $\iG$.
Observe that 
intervals in $\PersH(\filt)$ induce intervals in
$\PersH(\dec{\filt})$ in the following way:
\begin{itemize}
    \item For a $[b,d]\in\PersH(\filt)$
    and its sequence of complexes in $\filt$:
    \begin{equation}\label{eqn:seq-sw}
    K_{b},K_{b+1},\ldots,K_{d},
    \end{equation}
    let 
    \begin{equation}\label{eqn:seq-nonsw}
        \dec{K}_{\dec{b}}, \dec{K}_{\dec{b}+1}, \ldots, \dec{K}_{\dec{d}}
    \end{equation}
    be all the complexes in $\dec{\filt}$ appearing in 
    Sequence~(\ref{eqn:seq-sw}).
    If Sequence~(\ref{eqn:seq-nonsw}) is empty, then $[b,d]\in\PersH(\filt)$
    does not induce any interval in $\PersH(\dec{\filt})$;
    otherwise it induces an interval $[\dec{b},\dec{d}]\in\PersH(\dec{\filt})$.
\end{itemize}

Let $\{z_i\mid i\in[b,d]\}$ be the representative computed
for $[b,d]\in\PersH({\filt})$. 
Then, $\{z_{\indmap(\iG)}\mid \iG\in[\dec{b},\dec{d}]\}$
is a representative for $[\dec{b},\dec{d}]\in\PersH(\dec{\filt})$.



\section{Conclusion}
\revise{In this paper, we present an $O(m^2n)$ algorithm for computing representatives 
for a zigzag filtration $\filt$, where $m$ is the number of simplex insertions and deletions
in $\filt$ and $n$ is the maximum size of complexes in $\filt$.
The proposed algorithm improves the existing $O(m^2n^2)$ algorithm 
for computing zigzag representatives.
This improvement is achieved by maintaining the representatives 
in a more compact form using the new concept of wires and bundles.
An open question following our study is how to efficiently update the zigzag representatives
over local changes 
on the zigzag filtrations,
besides the existing works~\cite{dey2023revisiting,DBLP:conf/compgeom/DeyH24} 
on fast update of zigzag barcodes.
}

\opt{Arxiv}{\section*{Acknowledgment}
T.\ K.\ Dey was supported by NSF funds CCF 2049010 and DMS 2301360.
T.\ Hou was supported by NSF fund CCF 2439255.
D.\ Morozov
was supported by the U.S.\ Department of Energy, Office of Science, Office of Advanced Scientific Computing Research, Scientific Discovery through Advanced Computing (SciDAC) program, under Contract Number DE-AC02-05CH11231 at Lawrence Berkeley National Laboratory.}


\opt{Arxiv}{
\printbibliography
\appendix
}

\opt{SN}{

\section*{Declarations}
\begin{itemize}
    \item 
\textbf{Funding:} T.\ K.\ Dey was supported by National Science Foundation funds CCF 2049010, 2437030, and DMS 2301360.
T.\ Hou was supported by National Science Foundation fund CCF 2439255.
D.\ Morozov
was supported by the U.S.\ Department of Energy, Office of Science, Office of Advanced Scientific Computing Research, Scientific Discovery through Advanced Computing (SciDAC) program, under Contract Number DE-AC02-05CH11231 at Lawrence Berkeley National Laboratory.

\item
\textbf{Competing Interests:} The authors have no competing interests to declare that are relevant to the
content of this article.

\item
\textbf{Authors' contribution statements:} All authors (T.\ K.\ Dey, T.\ Hou, D.\ Morozov) contributed to the design of the approaches and the drafting of the manuscript. All authors reviewed and approved the current manuscript.
\end{itemize}

\bibliography{refs-SN}
}

\end{document}